\documentclass[10pt,onecolumn,draftclsnofoot]{IEEEtran}
\usepackage{tikz}
\usepackage{cite}
\usepackage{graphicx}
\usepackage{amsmath}
\usepackage{amsfonts}
\usepackage{amssymb}
\usepackage{mathrsfs}
\usepackage{fix2col}
\usepackage{latexsym}
\usepackage[dvips]{epsfig}
\usepackage[font=footnotesize]{subfig}
\usepackage{hyperref}
\usepackage{graphics}
\usepackage{epstopdf}
\usepackage{enumerate}
\usepackage{epstopdf}
\usepackage{amsthm}
\usepackage{arydshln}
\usepackage{mathdots}
\usepackage{MnSymbol}
\usepackage{wrapfig}

\makeatletter
\newcommand{\opnorm}{\@ifstar\@opnorms\@opnorm}
\newcommand{\@opnorms}[1]{%
	\left|\mkern-1.5mu\left|\mkern-1.5mu\left|
	#1
	\right|\mkern-1.5mu\right|\mkern-1.5mu\right|
}
\newcommand{\@opnorm}[2][]{%
	\mathopen{#1|\mkern-1.5mu#1|\mkern-1.5mu#1|}
	#2
	\mathclose{#1|\mkern-1.5mu#1|\mkern-1.5mu#1|}
}
\makeatother

\newtheorem{theorem}{Theorem}
\newtheorem{definition}{Definition}
\newtheorem{lemma}{Lemma}

\newtheorem{remark}{Remark}
\newtheorem{corollary}{Corollary}

\newtheorem{assumption}{Assumption}

\global\long\def\bl{\mathrm{block}}
\global\long\def\mA{\mathcal{A}}

\begin{document}

\title{  Sample Complexity of Sparse System\\ Identification Problem
	\thanks{Email: fattahi@berkeley.edu and sojoudi@berkeley.edu.}}
\author{
	\authorblockN{Salar Fattahi and Somayeh Sojoudi\\}
	\thanks{Salar Fattahi is with the Department of Industrial Engineering and Operations Research, University of California, Berkeley. Somayeh Sojoudi is with the Departments of Electrical Engineering and Computer Sciences and Mechanical Engineering as well as the Tsinghua-Berkeley Shenzhen Institute, University of California, Berkeley. This work was supported by the ONR Award N00014-18-1-2526 
and an NSF EPCN Grant. }}
\maketitle

\begin{abstract}
	In this paper, we study the system identification problem for sparse linear time-invariant systems. We propose a sparsity promoting block-regularized estimator to identify the dynamics of the system with only a limited number of input-state data samples. We characterize the properties of this estimator under high-dimensional scaling, where the growth rate of the system dimension is comparable to or even faster than that of the number of available sample trajectories. In particular, using contemporary results on high-dimensional statistics, we show that the proposed estimator results in a small element-wise error, provided that the number of sample trajectories is above a threshold. This threshold depends polynomially on the size of each block and the number of nonzero elements at different rows of input and state matrices, but only logarithmically on the system dimension. A by-product of this result is that  the number of  sample trajectories required for sparse system identification  is significantly smaller than the dimension of the system. Furthermore, we show that, unlike the recently celebrated least-squares estimators for system identification problems, the method developed in this work is capable of \textit{exact recovery} of the underlying sparsity structure of the system with the aforementioned number of data samples. Extensive case studies on synthetically generated systems, physical mass-spring networks, and multi-agent systems are offered to demonstrate the effectiveness of the proposed method.
\end{abstract}

\section{Introduction}
With their ever-growing size and complexity, real-world dynamical systems are hard to model. Today's systems are complex and large, often with a massive number of unknown parameters, which render them doomed to the so-called \textit{curse of dimensionality}. Therefore, system operators should rely on simple and tractable estimation methods to identify the dynamics of the system via a limited number of recorded input-output interactions, and then design control policies to ensure the desired behavior of the entire system. The area of \textit{system identification} is created to address this problem~\cite{ljung1998system}. 

Despite its long history in control theory, most of the results on system identification deal with the asymptotic behavior of various estimation methods~\cite{ljung1998system, ljung1978convergence, pintelon2012system, bai2010non}. Although these results shed light on the theoretical consistency of the existing methodologies, they are not applicable in the finite time/sample settings. In many applications, including neuroscience, transportation networks, and gene regulatory networks, the dimensionality of the system is overwhelmingly large, often surpassing the number of available input-output data~\cite{vertes2012simple, sun2012network, omranian2016gene}. Under such circumstances, the dynamics of the system should be estimated under the \textit{large dimension-small sample size} regime and  classical approaches for checking the asymptotic consistency of an estimator face major breakdowns. Simple examples of such failures are widespread in high-dimensional statistics. For instance, it is known that, given a number of independent and identically distributed (i.i.d.) samples with Gaussian distribution, none of the eigenvectors and eigenvalues of the sample covariance matrix are consistent estimators of their true counterparts if the sample size and the dimension of variables grow at the same rate~\cite{johnstone2001distribution}. As another example, it is well-known that the least-squares estimators, which are widely used in system identification problems, seize to exist uniquely when the sample size is smaller than the dimension of the system~\cite{cox1979theoretical}. 

Recently, a special attention has been devoted to  the \textit{sparse} system identification problem, where the states and inputs are assumed to possess localized or low-order interactions. These methods include, but are not restricted to, selective $\ell_{1}$-regularized estimator~\cite{6883140}, identification based on compressive sensing~\cite{sanandaji2014observability, jiang2014compressive}, sparse estimation of polynomial system dynamics~\cite{rojas2014sparse}, kernel-based regularization~\cite{chen2014system}, and low rank estimation in frequency domain~\cite{smith2014frequency}.
On the other hand, with the unprecedented interest in data-driven control approaches, such as model-free reinforcement learning, robust control, and adaptive control~\cite{ross2012agnostic,sadraddini2018formal, hou2011data}, a question arises as to what the minimum number of input-output data samples should be to guarantee a small error in the estimated model. Answering this question has been the subject of many recent studies on the sample complexity of the system identification problem~\cite{weyer1999finite, weyer2000finite, pereira2010learning, dean2017sample, tu2017non}. Most of these results are tailored to a specific type of dynamics, depend on the stability of the open-loop system, or do not exploit the \textit{a priori} information about the structure of the system. 

In this work, the objective is to employ modern results on high-dimensional statistics to reduce the sample complexity for one of the most fundamental classes of systems in control theory, namely  linear time-invariant (LTI) systems with perfect state measurements. This type of dynamical systems forms the basis of many classical control problems, such as Linear Quadratic Regulator and Linear Quadratic Gaussian problems. Our results are built upon the fact that, in many practical large-scale systems, the states and inputs exhibit sparse interactions with one another, which in turn translates into a sparse representation of the state-space equations of the system. Driven by the existing non-asymptotic results on the classical Lasso problem, the main focus of this paper is on the block-regularized estimators for the system identification problem, where the goal is to promote sparsity on different blocks of input and state matrices. To this goal, the $\ell_{\infty}$-norms of the blocks are penalized instead of their $\ell_{1}$-norms. One motivation behind employing this type of estimator stems from topology extraction in consensus networks, especially in the multi-agent setting~\cite{chernyshov2015towards, hassan2016topology}. In this problem, given a number of subsystems (agents) whose interactions are defined via an unknown sparse topology network, the objective is to estimate the state-space model governing the entire system based on a limited number of input-output sample trajectories. Since the subsystems have their own local state and input vectors with potentially different sizes, the parameters of the state-space model admit a block-sparse structure. 

One objective is to show that the developed estimator for recovering the block-sparsity of the true dynamics is guaranteed to achieve infinitesimal estimation error with a small number of samples. In particular, we use an $\ell_{1}/\ell_{\infty}$-regularized least-squares estimator, i.e., a least-squares estimator accompanied by a $\ell_{\infty}$ regularizer on different blocks, and show that, with an appropriate scaling of the regularization coefficient, $\Omega(k_{\max}(D\log(\bar{n}+\bar{m})+D^2))$ sample trajectories are enough to guarantee a small estimation error with a high probability, where $k_{\max}$ is the maximum number of nonzero elements in the rows of input and state matrices, $D$ is the size of the largest block in these matrices, and $\bar{n}$ and $\bar{m}$ are the number of row blocks in the state and input matrices, respectively. This is a significant improvement over the recently derived sample complexity of $\Omega(n+m)$ for the least-squares estimator ($n$ and $m$ are the state and input dimensions, respectively), in the case where the system is sparse and the sizes of all blocks are small relative to the system dimension. While the traditional Lasso is heavily studied in the literature~\cite{wainwright2009sharp, zhao2006model}, the high-dimensional behavior of the block-regularized estimator is less known when the dimensions of blocks are arbitrary. The paper~\cite{Martin11} analyzes the high-dimensional consistency of this estimator when each block of the regression parameter is a row vector. Furthermore, that work assumes that the regression parameter consists of only one column of blocks. In an effort to make these results applicable to the block-sparse system identification problem, we significantly generalize the existing non-asymptotic properties to problems with an arbitrary number of blocks, each with general sizes.

Moreover, we derive upper bounds on the element-wise error of the proposed estimator. In particular, we prove that $\Omega(k_{\max}^2(D\log(\bar{n}+\bar{m})+D^2))$ sample trajectories is enough to ensure that the estimation error decreases at the rate $O(\sqrt{({D\log(\bar{n}+\bar{m})+D^2})/{d}})$, where $d$ is the number of available sample trajectories. We show that if the number of nonzero elements in the columns (in addition to the rows) of input and state matrices are upper bounded by $k_{\max}$, the operator norm of the estimation error of the proposed estimator is \textit{arbitrarily smaller} than that of its un-regularized least-squares counterpart introduced in~\cite{dean2017sample}.
Another advantage of the proposed estimator over its least-squares analog is its \textit{exact recovery} property. More specifically, we show that while the least-squares estimator is unable to identify the sparsity pattern of the input and state matrices for \textit{any} finite number of samples, the proposed estimator recovers the true sparsity pattern of these matrices with a sublinear number of sample trajectories. It is worthwhile to mention that this work generalizes the results in~\cite{pereira2010learning}, where the authors use a similar regularized estimator to learn the dynamics of a particular type of systems. However,~\cite{pereira2010learning} ignores the block structure of the system and assumes autonomy and inherent stability, all of which will be relaxed in this work. To demonstrate the efficacy of the developed regularized estimator, three case studies are offered on synthetically generated systems, mass-spring networks, and multi-agent systems.

This work is a significant extension of our previous conference papers on Lasso-type estimators for system identification~\cite{fattahi2018data} and non-asymptotic analysis of block-regularized linear regression problems~\cite{fattahinon}. In particular, by combining the properties of the block-regularized regression and the characteristics of LTI systems, we provide a unified sparsity-promoting framework for estimating the parameters of the system with arbitrary block structures. To this goal, we have generalized our theoretical results in~\cite{fattahi2018data} and~\cite{fattahinon} to account for partially-sparse structures. Furthermore, we have relaxed certain assumptions on the structure of the true system that were initially required in~\cite{fattahi2018data}, and provided comprehensive discussions and more relevant simulations on the performance of the proposed method.

\vspace{2mm}

{\bf Notations:} For a matrix $M$, the symbols $\|M\|_F$, $\|M\|_2$, $\|M\|_1$, and $\|M\|_{\infty}$ denote its Frobenius, operator, $\ell_1/\ell_1$, and $\ell_{\infty}/\ell_{\infty}$ norms, respectively. Furthermore, $\kappa(M)$ refers to its 2-norm condition number, i.e., the ratio between its maximum and minimum singular values. Given integer sets $I$ and $J$, the notation $M_{IJ}$ refers to the submatrix of $M$ whose rows and columns are indexed by $I$ and $J$, respectively. Given the sequences $f_1(n)$ and $f_2(n)$, the notations $f_1(n) = O(f_2(n))$ and $f_1(n) = \Omega(f_2(n))$ imply that there exist $c_1<\infty$ and $c_2>0$ such that $f_1(n)\leq c_1f_2(n)$ and $f_1(n)\geq c_2f_2(n)$, respectively. Finally, $f_1(n) = o(f_2(n))$ is used to show that $f_1(n)/f_2(n)\rightarrow 0$ as $n\rightarrow\infty$. A zero-mean Gaussian distribution with covariance $\Sigma$ is shown as $N(0,\Sigma)$. Given a function $f(x)$, the expression $\arg\min f(x)$ refers to its minimizer. For a set $\mathcal{I}$, the symbol $|\mathcal{I}|$ denotes its cardinality.
\section{Problem Formulation}

Consider the LTI system
\begin{subequations}\label{ss}
	\begin{align}
		& x[t+1] = Ax[t] + Bu[t] + w[t]
	\end{align}
\end{subequations}
where $t$ is the time step, $A\in\mathbb{R}^{n\times n}$ is the state matrix, and $B\in\mathbb{R}^{n\times m}$ is the input matrix. Furthermore, $x[t]\in\mathbb{R}^n$, $u[t]\in\mathbb{R}^m$, and $w[t]\in\mathbb{R}^n$ are the state, input, and disturbance vectors at time $t$, respectively. The dimension of the system is defined as $m+n$. It is assumed that the input disturbance vectors are identically distributed and independent with distribution $N(0,\Sigma_w)$ across different times. In this work, we assume that the matrices $A$ and $B$ are sparse and the goal is to estimate them based on a limited number of \textit{sample trajectories}, i.e. a sequence $\{(x^{(i)}[\tau],u^{(i)}[\tau])\}_{\tau = 0}^T$ with $i = 1,2,...,d$, where $d$ is the number of available sample trajectories. The $i^{\text{th}}$ sample trajectory $\{(x^{(i)}[\tau],u^{(i)}[\tau])\}_{\tau = 0}^T$ is obtained by running the system from $t = 0$ to $t = T$ and collecting the input and state vectors. The sparsity assumption on $A$ and $B$ is practical in many applications because of two reasons:
\begin{itemize}
	\item In order to model many large-scale real-world systems accurately, one needs to consider an overwhelmingly large number of internal states. However, it is often the case that the interactions between different states and inputs obey a sparse structure, which translates into a sparse pattern in state and input matrices. An important example of these types of problems is multi-agent systems, where the agents (subsystems) interact with one another via a sparse communication network.
	\item More generally, it is well-known that there may not be a unique representation of the state-space model of the system. For example, one can entirely change the state and input matrices in~\eqref{ss} via linear/nonlinear transformations to arrive at a different, but equally accurate state-space equation. Furthermore, the dynamics of the system may have a sufficiently accurate sparse approximation. The question of interest is whether it is possible to design a data-driven method in order to estimate the \textit{sparsest} state-space representation of the system. Answering this question is crucial, specially in the context of a~\textit{distributed control} problem where the goal is to design a decentralized controller whose structure respects the dynamics of the system~\cite{darivianakishigh, wang2014localized, fazelnia2017convex}.
\end{itemize}

Given the sample trajectories $\{(x^{(i)}[\tau],u^{(i)}[\tau])\}_{\tau = 0}^\top$ for $i = 1,2,...,d$, one can obtain an estimate of $(A, B)$ by solving the following least-squares optimization problem:
\begin{equation}\label{ls}
\min_{A, B}\sum_{i=1}^d \sum_{t=0}^{T-1} \left\|x^{(i)}[t+1]-\left(Ax^{(i)}[t]+Bu^{(i)}[t]\right)\right\|_2^2 
\end{equation}
In order to describe the behavior of the least-squares estimator, define 
\begin{align}
 Y^{(i)} = \begin{bmatrix}
x^{(i)}[1]^\top\\
\vdots\\
x^{(i)}[T]^\top
\end{bmatrix},\ \ X^{(i)} = \begin{bmatrix}
x^{(i)}[0]^\top & u^{(i)}[0]^\top\\
\vdots & \vdots\\
x^{(i)}[T\!-\!1]^\top & u^{(i)}[T\!-\!1]^\top
\end{bmatrix}
,\ \ W^{(i)} = \begin{bmatrix}
w^{(i)}[0]^\top\\
\vdots\\
w^{(i)}[T-1]^\top
\end{bmatrix}.
\end{align}
for every sample trajectory $i = 1,2,...,d$. Furthermore, let $Y$, $X$, and $W$ be defined as vertical concatenations of $Y^{(i)}$, $X^{(i)}$, and $W^{(i)}$ for $i = 1,2,...,d$, respectively. Finally, denote $\Theta = [A\ \ B]^\top$ as the unknown system parameter and $\Theta^*$ as its true value. Based on these definitions, it follows from~\eqref{ss} that
\begin{equation}\label{compact}
Y = X\cdot\Theta + W
\end{equation} 
 The system identification problem is then reduced to estimating $\Theta$ based on the \textit{observation matrix} $Y$ and the \textit{design matrix} $X$. Consider the following least-squares estimator:
\begin{equation}\label{ls2}
{\Theta}_{\text{ls}} = \arg\min_{\Theta}\|Y-X\Theta\|_F^2
\end{equation} 
One can easily verify the equivalence of~\eqref{ls} and~\eqref{ls2}. The optimal solution of~\eqref{ls2} can be written as
\begin{equation}\label{estimator}
{\Theta}_{\text{ls}} = (X^\top X)^{-1}X^\top Y = \Theta^* + (X^\top X)^{-1}X^\top W
\end{equation}
Notice that ${\Theta}_{\text{ls}}$ is well-defined and unique if and only if $X^\top X$ is invertible, which necessitates $d\geq n+m$. The estimation error is then defined as 
\begin{equation}\label{errorls}
E = {\Theta}_{\text{ls}}-\Theta^* = (X^\top X)^{-1}X^\top W
\end{equation}
Thus, one needs to study the behavior of $(X^\top X)^{-1}X^\top W$ in order to control the estimation error of the least-squares estimator. However, since the state of the system at time $t$ is affected by random input disturbances at times $0,1,...t-1$, the matrices $X$ and $W$ are correlated, which renders~\eqref{errorls} hard to analyze. In order to circumvent this issue, \cite{Recht17} simplifies the estimator and considers only the state of the system at time $T$ in $Y^{(i)}$. By ignoring the first $T-1$ rows in $Y^{(i)}$, $X^{(i)}$, and $W^{(i)}$, one can ensure that the random matrix $(X^\top X)^{-1}X^\top$ is independent of $W$.
Therefore, it is assumed in the sequel that
\begin{align}\label{matrices}
 Y = \begin{bmatrix}
x^{(1)}[T]^\top\\
\vdots\\
x^{(d)}[T]^\top
\end{bmatrix},\ \  X = \begin{bmatrix}
x^{(1)}[T\!-\!1]^\top \!\!\!\!& u^{(1)}[T\!-\!1]^\top\\
\vdots & \vdots\\
x^{(d)}[T\!-\!1]^\top \!\!\!\!& u^{(d)}[T\!-\!1]^\top
\end{bmatrix},\ \  W = \begin{bmatrix}
w^{(1)}[T-1]^\top\\
\vdots\\
w^{(d)}[T-1]^\top
\end{bmatrix}\end{align}
With this simplification,~\cite{Recht17} shows that, with input vectors $u^{(i)}[t]$ chosen randomly from $N(0,\Sigma_u)$ for every $t = 1,2,...,T-1$ and $i = 1,2,...,d$, the least-squares estimator requires at least $d = \Omega(m+n+\log(1/\delta))$ sample trajectories to guarantee $\|E\|_2 = \mathcal{O}\left(\sqrt{{(m+n)\log(1/\delta)}/{d}}\right)$ with probability of at least $1-\delta$. In what follows, a block-regularized estimator will be introduced that exploits the underlying sparsity structure of the system dynamics to significantly reduce the number of sample trajectories for an accurate estimation of the parameters.
\section{Main Results}

Suppose that $A$ and $B$ can be partitioned as $A = [A^{(i,j)}]$ and $B = [B^{(k,l)}]$ where $(i,j)\in\{1,...,\bar{n}\}\times \{1,...,\bar{n}\}$ and $(k,l)\in\{1,...,\bar{n}\}\times \{1,...,\bar{m}\}$. $A^{(i,j)}$ is the $(i,j)^{\text{th}}$ block of $A$ with size $n_i\times n_j$. Similarly, $B^{(k,l)}$ is the $(k,l)^{\text{th}}$ block of $B$ with size $n_k\times m_l$.  Note that $\sum_{i = 1}^{\bar{n}}n_i = n$ and $\sum_{i = 1}^{\bar{m}}m_i = m$.
Suppose that it is known \textit{a priori} that all elements in each block $A^{(i,j)}$ or $B^{(k,l)}$ are simultaneously zero or nonzero. This implies that, as long as one element in $A^{(i,j)}$ or $B^{(k,l)}$ is nonzero, there is no reason to promote sparsity in the remaining elements of the corresponding block. 
Clearly, this kind of block-sparsity constraint is not correctly reflected in~\eqref{ls}. To streamline the presentation, we use the notation $\Theta = [A\ \ B]^\top$. Note that $\Theta^{(i,j)} = (A^{(j,i)})^\top$ for $i\in\{1,...,\bar{n}\}$ and $\Theta^{(i,j)} = (B^{(j,i-\bar{n})})^\top$ for $i\in\{\bar{n}+1,...,\bar{n}+\bar{m}\}$. In order to recover the true block-sparsity of $A$ and $B$, one can resort to an $\ell_1/\ell_\infty$ variant of the Lasso problem---known as the block-regularized least-squares (or simply block-regularized) problem:
\begin{equation}\label{lasso2}
\hat{\Theta} = \arg\min_{\Theta}\frac{1}{2d}\|Y-X\Theta\|_F^2 + \lambda_d\|\Theta\|_\bl
\end{equation}
where $\|\Theta\|_{\bl}$ is defined as the summation of $\|\Theta^{(i,j)}\|_\infty$ over $(i,j)\in\{1,...,\bar{n}+\bar{m}\}\times \{1,...,\bar{n}\}$. Suppose that $D$ denotes the maximum size of the blocks in $\Theta$, i.e., the multiplication of its number of rows and columns. Under the sparsity assumption on $(A,B)$, we will show that the non-asymptotic statistical properties of $\hat{\Theta}$ significantly outperform those of ${\Theta}_{\text{ls}}$. In particular, the primary objective is to prove that $\|\hat{\Theta}-\Theta^*\|_{\infty}$ decreases at the rate $\mathcal{O}(\sqrt{{D\log(n+m)+D^2\log(1/\delta)}/{{d}}})$ with probability of at least $1-\delta$ and with an appropriate scaling of the regularization coefficient, provided that $d = \Omega\left(k_{\max}^2\left(D\log(\bar{n}+\bar{m})+D^2\log(1/\delta)\right)\right)$. Here, $k_{\max}$ is the maximum number of nonzero elements in the columns of $[A\ \ B]^\top$. Comparing this number with the required lower bound $\Omega(n+m+\log(1/\delta))$ on the number of sample trajectories for the least-squares estimator, we conclude that the proposed method needs significantly less number of samples when $A$ and $B$ are sparse. The third objective is to prove that this method is able to find the correct block-sparsity structure of $A$ and $B$ with high probability. In contrast, it will be shown that the solution of the least-squares estimator is fully dense for any finite number of sample trajectories, and hence, it cannot correctly extract the sparsity structures of $A$ and $B$. We will showcase the superior performance of the block-regularized estimator both in sparsity identification and estimation accuracy in simulations.

To present the main results of this work, first note that
\begin{align}\label{output}
x^{(i)}[T\!-\!1] \!=\!& A^{T-2}Bu^{(i)}[0]\!+\!A^{T-3}Bu^{(i)}[1]\!+\!\cdots\!+\!Bu^{(i)}[T\!\!-\!\!2]\nonumber\\
&+A^{T-2}w^{(i)}[0]\!+\!A^{T-3}w^{(i)}[1]\!+\!\cdots\!+\!w^{(i)}[T\!\!-\!\!2]
\end{align}
where, without loss of generality, the initial state is assumed to be zero for every sample trajectory. The results can be readily extended to the case where the initial state is an unknown random vector with Gaussian distribution. Suppose that $u^{(i)}[t]$ and $w^{(i)}[t]$ are i.i.d samples of $N(0,\Sigma_u)$ and $N(0,\Sigma_w)$, respectively. Therefore,~\eqref{output} and~\eqref{matrices} imply that
\begin{equation}\label{dist}
X_{i,:}^\top \sim N\left(0, \tilde{\Sigma}\right)
\end{equation}
where $X_{i,:}$ is the $i^{\text{th}}$ row of $X$ and
\begin{subequations}
	\begin{align}
		& \tilde{\Sigma} = \begin{bmatrix}
		F_T\Sigma_uF_T^\top + G_T\Sigma_wG_T^\top & 0\\
		0 & \Sigma_u
		\end{bmatrix}\\
		& F_T = [A^{T-2}B\ \ A^{T-3}B\ \cdots \ B]\\
		& G_T = [A^{T-2}\ \ A^{T-3}\ \cdots \ I]
	\end{align}
\end{subequations}
Define $\mathcal{A}_j(\Theta) = \{i: \Theta^{(i,j)} \not = 0\}$. Unless stated otherwise, $\mathcal{A}_j$ is used to refer to $\mathcal{A}_j(\Theta^*)$. Define $\mathcal{A}^c_j$ as the complement of $\mathcal{A}_j$. For $\mathcal{T}\subseteq \{1,...,\bar{n}+\bar{m}\}$, denote ${I}(\mathcal{T})$ as the index set of rows in $\Theta^*$ corresponding to the blocks $\{{\Theta^*}^{(i,:)}: i\in\mathcal{T}\}$. For an index set $\mathcal{U}$, define $X_{\mathcal{U}}$ as a $d\times |\mathcal{U}|$ submatrix of $X$ after removing the columns with indices not belonging to $\mathcal{U}$. With a slight abuse of notation, $X_{(i)}$, $X_{\mathcal{A}_j}$, and $X_{\mathcal{A}_j^c}$ are used to denote $X_{{I}(\{i\})}$, $X_{{I(\mathcal{A}_j)}}$, and $X_{{I}(\mathcal{A}_j^c)}$ when there is no ambiguity. Similarly, $\Sigma_{(i),\mathcal{A}_j}$ and $\Sigma_{\mathcal{A}_j,\mathcal{A}_j}$ are used in lieu of $\Sigma_{{I}(\{i\}),I(\mathcal{A}_j)}$ and $\Sigma_{I(\mathcal{A}_j),I(\mathcal{A}_j)}$, respectively. Denote $k_j$ as the maximum number of nonzero elements in different columns of ${\Theta^*}^{(:,j)}$ which is the $j^{\text{th}}$ block column of $\Theta^*$. Finally, define 
\begin{align}
& n_{\max} = \max_{1\leq i\leq\bar{n}} n_i, && m_{\max} = \max_{1\leq i\leq\bar{m}} m_i,\nonumber\\ 
& p_{\max} = \max\left\{n_{\max}, m_{\max}\right\},
&& k_{\max} = \max_{1\leq j\leq\bar{n}} k_j,\nonumber\\
&\sigma^2_{\max} = \max_{1\leq i\leq n+m} \tilde\Sigma_{ii}
\end{align}
The following set of assumptions plays a key role in deriving the main result of this paper:
\begin{assumption}\label{asp:random}
	{\rm By fixing the time horizon $T$, the following conditions hold for all finite system dimensions:
		\begin{itemize}
			\item[A1.] (Mutual Incoherence Property): There exists a number $\gamma\in (0,1]$ such that 
			\begin{equation}
			\max_{j = 1,...,\bar{n}}\left\{\max_{i = 1,...,|\mathcal{A}_j^c|} \left\|\tilde\Sigma_{(i),\mathcal{A}_j}(\tilde\Sigma_{\mathcal{A}_j,\mathcal{A}_j})^{-1}\right\|_1\right\}\leq 1-\gamma
			\end{equation}
			\item[A2.] (Bounded eigenvalue): There exist numbers $0<\Lambda_{\min} < \infty$ and $0<\Lambda_{\max} < \infty$ such that
			\begin{equation}
			\Lambda_{\min}\leq \lambda_{\min}(\tilde\Sigma)\leq \lambda_{\max}(\tilde\Sigma)\leq \Lambda_{\max}
			\end{equation}
			\item[A3.](Bounded minimum value): There exists a number $t_{\min}>0$ such that
			\begin{equation}
				t_{\min}\leq\min_{1\leq j\leq\bar{n}}\min_{i\in\mathcal{A}_j}\left\|{\Theta^*}^{(i,j)}\right\|_{\infty}
			\end{equation}
			\item[A4.] (Block sizes): There exist numbers $\alpha_n, \alpha_m<\infty$ such that
			\begin{subequations}
			\begin{align}
				& n_{\max} = O\left((\bar{n}+\bar{m})^{\alpha_n}\right)\\
				& m_{\max} = O\left((\bar{n}+\bar{m})^{\alpha_m}\right)
			\end{align}
			\end{subequations}
	\end{itemize}}
\end{assumption}
The mutual incoherence property in Assumption A1 is a commonly known condition for the exact recovery of unknown parameters in compressive sensing and classical Lasso problems~\cite{zhao2006model, meinshausen2006high, donoho2006most, candes2007sparsity}. This assumption entails that the effect of those submatrices of $\tilde{\Sigma}$ corresponding to zero (unimportant) elements of $\Theta$ on the remaining entries of $\tilde{\Sigma}$ should not be large. Roughly speaking, this condition guarantees that the unknown parameters are \textit{recoverable} in the noiseless scenario, i.e. when $W = 0$. If the recovery cannot be guaranteed in the noise-free setting, then there is little hope for the block-regularized estimator to recover the true structure of $A$ and $B$ when the system is subject to noise. This assumption is satisfied in all of our simulations.

The bounded eigenvalue condition in Assumption A2 entails that the condition number of $\tilde{\Sigma}$ is bounded away from $0$ and $\infty$ for all finite system dimensions. However, as will be shown later, the values of $\Lambda_{\min}$ and $\Lambda_{\max}$ can change with respect to the time horizon $T$. In particular, it will be later shown  that for highly unstable systems, $\tilde{\Sigma}$ becomes severely ill-conditioned as the time horizon increases, which in turn makes the system identification problem difficult to solve. Furthermore, this assumption implies that there exists a constant $\bar\sigma_{\max}^2<\infty$ such that $\max_{1\leq i\leq n+m} \Sigma_{ii}\leq \bar\sigma_{\max}^2$ for every finite system dimension.

Assumption A3 implies that, independent of the system dimensions, there always exists a strictly positive gap between the zero and nonzero elements of $A$ and $B$. This assumption holds in almost all practical settings and will facilitate the exact sparsity recovery of the parameters of the system. It is worthwhile to mention that this condition can be relaxed to the case where $t_{\min}$ is a decreasing function of the system dimension, provided that its decrease rate is \textit{slow enough}. To streamline the presentation, we did not consider this case in the paper. 
 
 Finally, Assumption A4 requires that the maximum size of the blocks in $\Theta^*$ be polynomially bounded by the number of its block columns. 
 For instance, $\bar{n} = O(1)$ and $\bar{m} = O(1)$ violate this assumption since it implies that $n_{\max} = \Omega((\bar{n}+\bar{m})^{\log n})$ and $m_{\max} = \Omega((\bar{n}+\bar{m})^{\log m})$. It is worthwhile to mention that Assumption A4 results in $k_{\max} = O((\bar{n}+\bar{m})^{\alpha_k})$ for some number $\alpha_k<\infty$; this will be used later in the paper.
 
 \begin{remark}
 	Note that, due to  Assumption A2, $\kappa(\tilde{\Sigma}) = O(1)$. However, this quantity will not be removed from the big-$O$ analysis of our subsequent theorems and corollaries to demonstrate its effect on the high-dimensional properties of the developed estimator.
 \end{remark}
 
 The first theorem characterizes the non-asymptotic behavior of the Lasso-type estimator for the system identification problem where each block has the size $1\times 1$.

\begin{theorem}[element-wise regularization]\label{thm1}
	Suppose that $n_{\max} = m_{\max} = 1$. Upon choosing
	\begin{subequations}
	\begin{align}
	& \lambda_d = \Theta\left(\sigma_{\max}\sqrt{\frac{\log(n+m)+\log({1}/{\delta})}{d}}\right)\label{lambda}\\
	& d = \Omega\left(\kappa(\tilde{\Sigma})^2k_{\max}\left(\log(n+m)+\log({1}/{\delta})\right)\right)\label{d}
	\end{align}
	\end{subequations}
	the following statements hold with probability of at least $1-\delta$:
	\begin{itemize}
		\item[1.] $\hat{\Theta}$ is unique and has the same nonzero elements as $\Theta^*$.
		\item[2.] We have
		\begin{align}\label{error}
		&\hspace{-0.7cm} \|\hat{\Theta}\!-\!\Theta^*\|_{\infty}\!=\! O\Bigg(\kappa(\tilde{\Sigma})\left(1\!+\!\sqrt{\frac{k_{\max}(k_{\max}+\log({n}+{m})+\log(1/\delta))}{d}}\right)\sqrt{\frac{\log({n}+{m})+\log(1/\delta)}{d}}\Bigg)
		\end{align}
	\end{itemize}
\end{theorem}

Theorem~\ref{thm1} states that, for sparse system dynamics,  the support of $\hat\Theta$ will be recovered exactly, provided that the number of samples exceeds a threshold that is only logarithmic in the system dimension. This is specifically useful for large-scale systems where $k_{\max}$ is significantly smaller than $n+m$. Intuitively, this corresponds to those systems whose states are affected only by a few, and possibly local, states and inputs.

Define $D = p_{\max}n_{\max}$, which is the maximum size of the blocks in $\Theta$. The next theorem extends the result of Theorem~\ref{thm1} to arbitrary block sizes.
  
\begin{theorem}[block-wise regularization]\label{thm2}
	Upon choosing
	\begin{subequations}
	\begin{align}
	& \lambda_d = \Theta\left(\!\sigma_{\max}\sqrt{\frac{D\log(\bar{n}+\bar{m})+D^2\log(1/\delta)}{d}}\right)\label{lambda2}\\
	& d = \Omega\left(\kappa(\tilde{\Sigma})^2k_{\max}\left(D\log(\bar{n}+\bar{m})+D^2\log(1/\delta)\right)\right)\label{d22}
	\end{align}
	\end{subequations}
	the following statements hold with probability of at least $1-\delta$:
	\begin{itemize}
		\item[1.] $\hat{\Theta}$ is unique and has the same nonzero blocks as $\Theta^*$.
		\item[2.] We have
		\begin{align}\label{error2}
		&\hspace{-0.7cm}g = \|\hat{\Theta}\!-\!\Theta^*\|_{\infty}\!=\! O\Bigg(\kappa(\tilde{\Sigma})\left(1\!+\!\sqrt{\frac{k_{\max}(k_{\max}n_{\max}\!+\!\log(\bar{n}+\bar{m})\!+\!\log(1/\delta))}{d}}\right)\sqrt{\frac{D\log(\bar{n}\!+\!\bar{m})\!+\!D^2\log(1/\delta)}{d}}\Bigg)
		\end{align}
	\end{itemize}
\end{theorem}

Theorem~\ref{thm2} shows that the minimum number of required sample trajectories is a quadratic function of the maximum block size. Therefore, only a small number of samples is enough to guarantee the uniqueness, exact block-sparsity recovery, and small estimation error for sparse systems, assuming that the sizes of the blocks are significantly smaller than the system dimensions. 

Another important observation can be made from Theorems~\ref{thm1} and~\ref{thm2}. Notice that the minimum number of required sample trajectories and the element-wise error of the estimated parameter depends on $\kappa(\tilde{\Sigma})$. Upon assuming that $\Sigma_u$ and $\Sigma_w$ are identity matrices, $\kappa(\tilde{\Sigma})$ reduces to $\kappa(F_TF_T^\top+G_TG_T^\top)$. $F_T$ and $G_T$ are commonly known as finite-time \textit{controllability matrices} for the input and disturbance noise, respectively. Roughly speaking, $\kappa(F_TF_T^\top+G_TG_T^\top)$ quantifies the ratio between the eigenvalues of the easiest- and hardest-to-identify modes of the system. Therefore, Theorems~\ref{thm1} and~\ref{thm2} imply that only a small number of samples is required to accurately identify the dynamics of the system if all of its modes are easily excitable. The dependency of the estimation error on the modes of the system is also reflected in the non-asymptotic error bound of the least-squares estimator in~\cite{Recht17}. This is completely in line with the conventional results on the identifiability of dynamical systems: independent of the method in use, it is significantly harder to identify the parameters of the system accurately if it possesses nearly-hidden modes. The connection between the identifiability of the system and the number of required sample trajectories to guarantee a small estimation error will be elaborated through different case studies in Section~\ref{sec:num}.

In what follows, the effect of different scalings of block sizes on the minimum number of required sample trajectories and the element-wise error of the proposed estimator will be elaborated. In particular, it will be shown that, as the size of the blocks grows, a proportional increase in the number of sample trajectories is required for the successful recovery of the system parameters.

\begin{corollary}\label{cor1}
	Assume that $n_{\max} = O(1)$ and $m_{\max} = O(1)$. Then, the block-regularized estimator behaves similarly to the element-wise regularized estimator, i.e., the respective choices of~\eqref{lambda} and~\eqref{d} for $\lambda_d$ and $d$ are enough to guarantee exact block-sparsity recovery of $\Theta^*$ and the element-wise error of~\eqref{error} with probability of at least $1-\delta$.
\end{corollary}
\begin{proof}
	The proof is immediate by noting that $D = O(1)$ and $O(\log(\bar{n}+\bar{m})) = O(\log({n}+{m}))$.
\end{proof}

Corollary~\ref{cor1} shows that the block-regularized estimator behaves similarly to its element-wise counterpart if the sizes of the blocks do not change as the system dimension grows.

\begin{corollary}\label{cor2}
	Assume that $n_{\max} = O(\log n)$ and $m_{\max} = O(\log m)$. Then,
	\begin{subequations}
	\begin{align}
	&\lambda_d = \Theta\left(\!\sigma_{\max}\log^2({n}+{m})\sqrt{\frac{\log(1/\delta)}{d}}\right)\label{lambdadcor1}\\
	&d = \Omega(\kappa(\tilde{\Sigma})^2k_{\max}^2\log^4({n}+{m})\log(1/\delta))\label{dcor2}
	\end{align}
	\end{subequations}
	is enough to guarantee the exact block-sparsity recovery of $\Theta^*$ and
	\begin{equation}\label{errorcor1}
	\|\hat{\Theta}\!-\!\Theta^*\|_{\infty} = O\left(\kappa(\tilde{\Sigma})\log^2({n}+{m})\sqrt{\frac{\log(1/\delta)}{d}}\right)
	\end{equation}
	with probability of at least $1-\delta$.
\end{corollary}
\begin{proof}
	First, note that $D = O(\log^2(n+m))$. Therefore,~\eqref{lambda2} is equivalent to~\eqref{lambdadcor1}. Furthermore,~\eqref{dcor2} implies that 
	\begin{equation}\label{o1}
		\sqrt{\frac{k_{\max}(k_{\max}n_{\max}\!+\!\log(\bar{n}+\bar{m})\!+\!\log(1/\delta))}{d}} = O(1)
	\end{equation}
	This reduces~\eqref{error2} to~\eqref{errorcor1}, thereby completing the proof.
\end{proof}
Corollary~\ref{cor2} studies the behavior of the proposed estimator when the sizes of the blocks increase logarithmically in the system dimensions. The effect of such growth in the size of the blocks is reflected in the minimum number of necessary sample trajectories. However, it can be seen that under the scenarios where $k_{\max}$ is small, i.e. when the system is sparse, this number can still be significantly smaller than the system dimensions.

\begin{corollary}\label{cor3}
	Assume that $n_{\max} = O(n^{\beta_n})$ and $m_{\max} = O(m^{\beta_m})$ for some $\beta_n>0$ and $\beta_m>0$. Then,
	\begin{subequations}
	\begin{align}
	&\lambda_d = \Theta\left(\!\sigma_{\max}({n}+{m})^{\left(\beta_n+\beta_m\right)}\sqrt{\frac{\log(1/\delta)}{d}}\right)\\
	&d = \Omega(\kappa(\tilde{\Sigma})^2k_{\max}^2({n}+{m})^{2\left(\beta_n+\beta_m\right)}\log(1/\delta))
	\end{align}
	\end{subequations}
	is enough to guarantee the exact sparsity recovery of $\Theta^*$ and
	\begin{equation}
	\|\hat{\Theta}\!-\!\Theta^*\|_{\infty} = O\left(\kappa(\tilde{\Sigma})({n}+{m})^{\left(\beta_n+\beta_m\right)}\sqrt{\frac{\log(1/\delta)}{d}}\right)
	\end{equation}
	with probability of at least $1-\delta$.
\end{corollary}

\begin{proof}
	The proof is similar to that of Corollary~\ref{cor2}.
\end{proof}

Corollary~\ref{cor3} is the analog of Corollaries~\ref{cor1} and~\ref{cor2} for the \textit{polynomial scaling} of the block size. Similar to the previous cases, it can be seen that the size of the required sample trajectories heavily depends on the growth rate of the maximum block size of $\Theta$. Although the sampling rate is still sublinear when $\beta_n+\beta_m<1/2$, it may surpass the system dimension if $\beta_n+\beta_m>1/2$. A question arises as to whether one can resort to the ordinary least-squares estimator in lieu of the proposed block-regularized estimator for the cases where $\beta_n+\beta_m>1/2$ since the proposed estimator requires $d = \Omega((n+m)^{1+\epsilon}\log(1/\delta))$ for some $\epsilon>0$ whereas $d = \Theta(n+m+\log(1/\delta))$ in enough to guarantee the uniqueness of the least-squares estimator. However, in what follows, we will prove that the least-squares estimator does not extract the correct sparsity structure of $\Theta$ for \textit{any finite number} of sample trajectories.
\begin{theorem}
	If $A$ and $B$ are not fully dense matrices, ${\Theta}_{ls}$ does not recover the support of $\Theta^*$ for any finite number of sample trajectories with probability 1.
\end{theorem}
\begin{proof}
	Define $R = ((X^\top X)^{-1}X^\top)^\top$, and note that $R$ and $W$ are independent random variables due to the construction of $X$. Now, suppose that $\Theta_{ij}^* = 0$. We show that, with probability zero, $E_{ij} = |(\Theta_{\text{ls}})_{ij}-\Theta_{ij}^*| = 0$ holds. Note that $E_{ij} = R_{:, i}^\top W_{:,j}$. If $R_{:, i}\not= 0$, then $E_{ij}$ is a linear combination (with at least one nonzero coefficient) of identically distributed normal random variables with mean zero and variance $(\Sigma_w)_{jj}$. Since $R_{:,i}$ and $W_{:,j}$ are independent, we have $E_{ij} = 0$ with probability zero. Now, assume that $R_{:, i} = 0$. This means that the $i^{\text{th}}$ row of $R^\top$ is a zero vector. This, in turn, implies that the $i^{\text{th}}$ row of $R^\top X$ is zero. However, $R^\top X = (X^\top X)^{-1}X^\top X = I$, which is a contradiction. This completes the proof. 
\end{proof}

Define $h(n,m)=\sqrt{{(n+m)\log(1/\delta)}/{d}}$ and recall that $\|\Theta_{\mathrm{ls}}-\Theta^*\|_2 = O(h(n,m))$. In the next corollary, we show that, under additional sparsity conditions, the operator norm of the estimation error for $\hat\Theta$ becomes arbitrarily smaller than $h(n,m)$ as the system dimension grows. 

\begin{corollary}\label{cor4}
	Assume that the number of nonzero elements at different rows and columns of $\Theta^*$ is upper bounded by $k_{\max}$. Furthermore, suppose that $\lambda_d$ satisfies~\eqref{lambda2} and
	\begin{equation}\label{d4}
		d = \Omega\left(\kappa(\tilde{\Sigma})^2k_{\max}^2\left(D\log(\bar{n}+\bar{m})+D^2\log(1/\delta)\right)\right)
	\end{equation}
	Then, we have
	\begin{equation}\label{error4}
		\|\hat{\Theta}\!-\!\Theta^*\|_2 \!=\! O\Bigg(\underbrace{\kappa(\tilde{\Sigma})k_{\max}\sqrt{\frac{D\log(\bar{n}\!+\!\bar{m})\!+\!D^2\log(1/\delta)}{d}}}_{v(n,m)}\Bigg)
	\end{equation}
	with probability of at least $1-\delta$. Furthermore, we have
	\begin{equation}\label{hv}
		\frac{v(n,m)}{h(n,m)}\rightarrow 0\quad \textit{as}\quad (n,m)\rightarrow\infty
	\end{equation}
	provided that
	\begin{equation}\label{kD}
		k_{\max}D = o\left(\sqrt{\frac{n+m}{\log(n+m)}}\right)
	\end{equation}
\end{corollary}
\begin{proof}
	One can use Holder's inequality to write
	\begin{equation}\label{holder}
		\|\hat{\Theta}\!-\!\Theta^*\|_2\leq \sqrt{\|\hat{\Theta}\!-\!\Theta^*\|_1 \|\hat{\Theta}\!-\!\Theta^*\|_\infty}\leq k_{\max}\|\hat{\Theta}\!-\!\Theta^*\|_{\infty}
	\end{equation}
	On the other hand,~\eqref{o1} can be verified under~\eqref{d4}. Combined with~\eqref{holder} and Theorem~\ref{thm2}, this certifies the validity of~\eqref{error4}. It remains to prove the correctness of~\eqref{hv}. Note that under~\eqref{kD}, we have
	\begin{subequations}
	\begin{align}
		& k_{\max}^2D\log(\bar{n}+\bar{m}) = o\left(n+m\right)\\
		& k_{\max}^2D^2 = o\left(n+m\right)
	\end{align}
	\end{subequations}
	Combined with the definitions of $h(n,m)$ and $v(n,m)$, this completes the proof.
\end{proof}

Corollary~\ref{cor4} describes the settings under which our proposed method significantly outperforms the least-squares estimator in terms of the operator norm of the errors. This improvement is more evident for those systems where the states and inputs have sparse interactions and the block sizes in $A$ and $B$ are smaller than the system dimensions. A class of such systems is multi-agent networks where the agents interact only locally and their total number dominates the dimension of each individual agent.

\section{Proofs}

In this section, the proof of Theorem~\ref{thm2} will be presented. Furthermore, it will be shown that Theorem~\ref{thm1} is a special case of this theorem. A number of preliminary definitions and lemmas are required to present the proof of Theorem~\ref{thm2}. 

\begin{definition}[sub-Gaussian random variable]
	{\rm A zero-mean random variable $x$ is \textit{sub-Gaussian} with parameter $\sigma^2$ if there exists a constant number $c<\infty$ such that
		\begin{equation}
		\mathbb{P}(|x|>t)\leq c\cdot\exp\left(-\frac{t^2}{2\sigma^2}\right)
		\end{equation}}
\end{definition}

\begin{lemma}\label{l1}
	Given a set of zero-mean sub-Gaussian variables $x_i$ with parameters $\sigma_i$ for $i = 1,2,...,m$, the inequality
	\begin{equation}
	\mathbb{P}\left(\max_{i}|x_i|>t\right)\leq c\cdot\exp\left(-\frac{t^2}{2\max_i \sigma_i^2}+\log m\right)
	\end{equation}
	holds for some constant $c<\infty$.
\end{lemma}
Define $I_d$ as the $d\times d$ identity matrix. The next two lemmas are borrowed from~\cite{Martin11} and~\cite{wainwright2009sharp}, respectively.
\begin{lemma}\label{l2}
	Given a set of random vectors $X_i\sim N(0,\sigma_i^2I_d)$ for $i = 1,2,...,m$ and $d> 2\log m$, the inequality
	\begin{equation}
	\mathbb{P}\left(\max_i\|X_i\|_2^2\geq 4\sigma^2 d\right)\leq \exp\left(-\frac{d}{2}+\log m\right)
	\end{equation}
	holds, where $\sigma = \max_i \sigma_i$.
\end{lemma}


\begin{lemma}\label{l3}
	Consider a matrix $X\in\mathbb{R}^{m\times n}$ whose rows are drawn from $N(0, \Sigma)$. Assuming that $n\leq m$, we have
	\begin{equation}
	\mathbb{P}\left(\left\|\left(\frac{1}{d}X^\top X\right)^{-1}-\Sigma^{-1}\right\|_2\geq \frac{8}{\Lambda_{\min}}\sqrt{\frac{t}{m}}\right)\leq 2\exp\left(-\frac{t}{2}\right)
	\end{equation}
	for every $n\leq t\leq m$.
\end{lemma}

The basic inequalities given below will be used frequently in our subsequent arguments.
\begin{lemma}\label{l4}
	The following statements hold true:
	\begin{itemize}
		\item Given a number of (not necessarily independent) events $\mathcal{T}_i$ for $i = 1,2,...,n$, the following inequality is satisfied:
		\begin{equation}
		\sum_{i = 1}^{n} \mathbb{P}(\mathcal{T}_i)-(n-1)\leq \mathbb{P}(\mathcal{T}_1\cap \mathcal{T}_2\cap...\cap \mathcal{T}_n)
		\end{equation}
		\item Given events $\mathcal{B}$ and $\mathcal{C}$ together with the complement of $\mathcal{C}$, denoted as $\mathcal{C}^c$, the following inequality holds:
		\begin{equation}
		\mathbb{P}(\mathcal{B})\leq \mathbb{P}(\mathcal{B}|\mathcal{C})+\mathbb{P}(\mathcal{C}^c)
		\end{equation}
	\end{itemize}
\end{lemma}

The next lemma characterizes the first-order optimality conditions for~\eqref{lasso2}.
\begin{lemma}[KKT conditions]\label{kkt}
	$\hat{\Theta}$ is an optimal solution for~\eqref{lasso2} if and only if it satisfies
	\begin{equation}\label{eqkkt}
	\frac{1}{d}X^\top X(\hat{\Theta}-\Theta^*)-\frac{1}{d}X^\top W+\lambda_d\hat{S} = 0
	\end{equation}
	for some $\hat{S}\in\mathbb{R}^{(n+m)\times n}\in \partial\|\hat{\Theta}\|_{\bl}$, where $\partial\|\hat{\Theta}\|_{\bl}$ denotes the sub-differential of $\|\cdot\|_{\bl}$ at $\hat{\Theta}$.
\end{lemma}
\begin{proof}
	The proof is straightforward and omitted for brevity.~\end{proof}
$\hat{S}_{\mA}$ and $\hat{S}_{\mA^c}$ are obtained by removing those blocks of $\hat{S}$ with indices not belonging to $\mA$ and $\mA^c$, respectively. 
The equation
\eqref{compact} can be reformulated as the set of linear equations
\begin{equation}\label{lasso_small}
Y^{(:,j)} = X\Theta^{(:,j)}+W^{(:,j)}\quad \forall j \in\{1,...,n\}
\end{equation}
where $Y^{(:,j)}$, $\Theta^{(:,j)}$, and $W^{(:,j)}$ are the $j^{\text{th}}$ block column of $Y$, $\Theta$, and $W$, respectively. Based on this definition, consider the following set of block-regularized subproblems:
\begin{equation}\label{lasso3}
{\hat\Theta}^{(:,j)} = \arg\min\frac{1}{2d}\|Y^{(:,j)}-X\Theta^{(:,j)}\|_2^2 + \lambda_d\|\Theta^{(:,j)}\|_\bl
\end{equation}
Define $D_j = p_{\max}n_j$.
The next two lemmas are at the core of our proof for Theorem~\ref{thm2}.

\begin{lemma}[No false positives]\label{sparsity}
	Given arbitrary constants $c_1,c_2>1$, suppose that $\lambda_d$ and $d$ are chosen such that
	\begin{subequations}
	\begin{align}
	& \lambda_d\geq \sqrt{\frac{32c_1 \sigma_w^2\sigma_{\max}^2}{\gamma^2}\cdot\frac{(D_j)^2+D_j\log (\bar{n}+\bar{m})}{d}}\label{lambda3}\\
	& d\geq \frac{72c_2\sigma_{\max}^2}{\gamma^2\Lambda_{\min}}\cdot k_j (D_j^2+D_j\log(\bar{n}+\bar{m}))\label{d3}
	\end{align}
	\end{subequations}
	Then, with probability of at least
	\begin{equation}\label{p3}
	\begin{aligned}
	1-3&\exp\big(-(c_1-1)(D_j+\log(\bar{n}+\bar{m}))\big)-4\exp\big(-(c_2-1)(D_j+\log(\bar{n}+\bar{m}))\big)
	\end{aligned}
	\end{equation}
	$\hat{\Theta}^{(:,j)}$ is unique and its nonzero blocks exclude the zero blocks of ${\Theta^*}^{(:,j)}$. In other words, $\hat{\Theta}^{(:,j)}$ does not have any false positives.
\end{lemma}

Recall that due to Assumption A4, one can write $n_{\max} = O\left((\bar{n}+\bar{m})^{\alpha_n}\right)$ and $k_{\max} = O\left((\bar{n}+\bar{m})^{\alpha_k}\right)$ for some $\alpha_n\geq 0$ and $\alpha_k\geq 0$.
\begin{lemma}[Element-wise error]\label{errorbound}
	Given  arbitrary constants $c_3>0$ and $c_4>1$, suppose that $\hat{\Theta}$ is unique and the set of its nonzero blocks excludes the zero blocks of $\Theta^*$. 
	Then, with probability of at least
	\begin{equation}\label{p4}
	\begin{aligned}
	1&-2\exp(-({k_jn_j+c_3\log(\bar{n}+\bar{m})})/{2})-2\exp\left(-d/2\right)-2\exp\big(-2(c_4-1)(\alpha_n+\alpha_k)\log(\bar{n}+\bar{m}))\big)
	\end{aligned}
	\end{equation}
	we have
	\begin{align}\label{element_error}
	\|\hat{\Theta}^{(:,j)}\!-\!{\Theta^*}^{(:,j)}\|_{\infty}\leq \sqrt{\frac{36c_4(\alpha_n\!+\!\alpha_k)\sigma_w^2\log(\bar{n}+\bar{m})}{\Lambda_{\min} d}}
	+\frac{\lambda_d}{\Lambda_{\min}}\left({8\sqrt{k_j}\sqrt{\frac{k_jn_j+c_3\log(\bar{n}+\bar{m})}{d}}}+1\right)
	 \!=\! g_j
	\end{align}
	Furthermore, the zero blocks of $\hat{\Theta}^{(:,j)}$ exclude the nonzero blocks of ${\Theta^*}^{(:,j)}$ if $\min_{i\in\mathcal{A}_j}\|\Theta^{(i,j)}\|_{\infty}> g_j$. In other words, $\hat{\Theta}^{(:,j)}$ does not have any false negatives if $\min_{i\in\mathcal{A}_j}\|\Theta^{(i,j)}\|_{\infty}> g_j$.
\end{lemma}

In what follows, we will present some preliminaries that are essential in proving Lemmas~\ref{sparsity} and~\ref{errorbound}.
Notice that $\hat{S}$ and $W$ have the same dimensions as $\hat{\Theta}$, and hence, can be similarly partitioned into different blocks.
Since Lemmas~\ref{sparsity} and~\ref{errorbound} hold for any given column block index $j$, $\Theta^{(i,j)}$ and $\mA_j$ will be referred to as $\Theta^{(i)}$ and $\mA$ in order to streamline the presentation.
\begin{lemma}\label{subdiff}
	$Q\in \partial\|\tilde{\Theta}\|_{\bl}$ if and only if the following conditions are satisfied for every $i \in\{1,2,..., \bar{n}+\bar{m}\}$:
	\begin{itemize}
		\item If $\|\tilde{\Theta}^{(i)}\|_{\infty} \not= 0$, define $M^{(i)} = \{(k,l): \tilde{\Theta}^{(i)}_{kl} = \|\tilde{\Theta}^{(i)}\|_{\infty} \}$. Then, $Q^{(i)}_{kl} = \eta_{kl}\cdot\mathrm{sign}(\tilde{\Theta}^{(i)}_{kl})$, where $\sum_{(k,l)\in M^{(i)}}\eta_{kl} = 1$ and $\eta_{kl} = 0$ if $(k,l)\not\in M^{(i)}$.
		\item If $\|\tilde{\Theta}^{(i)}\|_{\infty} = 0$, then $\|Q^{(i)}\|_1\leq 1$.
	\end{itemize}
\end{lemma}

The proofs of Lemmas~\ref{sparsity} and~\ref{errorbound} are based on the well-known primal-dual witness approach introduced in~\cite{wainwright2009sharp,Martin11}, which is defined as follows:

\vspace{2mm}
\hrule
\hrule
\vspace{2mm}
{\bf Primal-dual witness approach (\cite{wainwright2009sharp,Martin11}):}

\begin{itemize}
	\item \textit{Step 1:} Define the restricted regularized problem as
	\begin{subequations}\label{restricted}
		\begin{align}
		\tilde{\Theta} = & \arg\min_{\Theta\in\mathbb{R}^{p\times r}} && \frac{1}{2d}\|Y-X\Theta\|_F^2+\lambda_d\|\Theta\|_{\bl}\\
		& \mathrm{s.t.} && \Theta^{(i)} = 0\quad \forall i\in\mathcal{A}^{c}
		\end{align}
	\end{subequations}
	whose solution is unique if $X_{\mathcal{A}}^\top X_{\mathcal{A}}$ is invertible.
	\item \textit{Step 2:} With a slight abuse of notation, $\tilde{\Theta}$ can be written as $(\tilde{\Theta}_{\mathcal{A}}, 0)$. Choose $\tilde{S}_{\mA}$ as an element of the sub-differential $\partial\|\tilde{\Theta}_{\mA}\|_{\bl}$.
	\item \textit{Step 3:} Find $\tilde{S}_{\mA}^c$ by solving the KKT equations~\eqref{eqkkt}, given $\tilde{\Theta}$ and $\tilde{S}_{\mA}$. Then, verify 
	\begin{equation}\label{eqPDW}
	\|\tilde{S}^{(i)}\|_1<1\quad\forall i\in\mathcal{A}^c
	\end{equation}
\end{itemize}
\vspace{0mm}
\hrule
\hrule
\vspace{2mm}

If~\eqref{eqPDW} can be verified in the last step, it is said that the primal-dual witness (PDW) approach \textit{succeeds}.
The next lemma unveils a close relationship between the block-regularized estimator, the PDW approach, and the true regression parameter $\Theta^*$. 
\begin{lemma}\label{pdw}
	The following statements hold:
	\begin{itemize}
		\item If the PDW approach succeeds, then $\tilde{\Theta}$ is the unique optimal solution of~\eqref{lasso2}, i.e. $\hat{\Theta} = \tilde{\Theta}$. 
		\item Conversely, suppose that $\hat{\Theta}$ is the optimal solution of~\eqref{lasso2} such that $\hat{\Theta}^{(i)} = 0$ for every $i\in\mathcal{A}^c$. Then, the PDW approach succeeds.
	\end{itemize}
\end{lemma}
\begin{proof}
	The proof is a simple generalization of Lemma 2 in~\cite{Martin11}. The details are omitted for brevity.
\end{proof}
Lemma~\ref{pdw} is the building block of our proofs for Lemmas~\ref{sparsity} and~\ref{errorbound}. In particular, Lemma~\ref{pdw} indicates that in order to show that the solution of~\eqref{lasso_small} is unique and excludes false positive errors, it is enough to verify that the PDW approach succeeds with high probability. Then, conditioned on the success of the PDW approach, our focus can be devoted to the optimal solution of the restricted problem~\eqref{restricted} and bounding its difference from the true parameters. 

\begin{lemma}\label{lerror}
	Define $\tilde{\Theta}-\Theta = E$. The following equalities hold:
	\begin{subequations}
	\begin{align}
	& \tilde{S}_{\mA^c} = 0\label{eqAc}\\
	& E_{\mA} = (\frac{1}{d}X_{\mA}^\top X_{\mA})^{-1}\frac{1}{d}X_{\mA}^\top W-(\frac{1}{d}X_{\mA}^\top X_{\mA})^{-1}\lambda_d\tilde{S}_{\mA}\label{eqA}\\
	& \tilde{S}_{\mA^c} = \frac{1}{d\lambda_d}\left(X_{\mA^c}^\top-(X_{\mA^c}^\top X_{\mA})(X_{\mA}^\top X_{\mA})^{-1}X_{\mA}^\top\right)W\nonumber\\
	&\hspace{8mm}+\frac{1}{d}X_{\mA^c}^\top X_{\mA}(\frac{1}{d}X_{\mA}^\top X_{\mA})^{-1}\tilde{S}_{\mA}\label{eqAcs}
	\end{align}
	\end{subequations}
\end{lemma}

\begin{proof}
	To verify~\eqref{eqA} and~\eqref{eqAcs}, note that the KKT condition in Lemma~\ref{kkt} reduces to
	\begin{subequations}
	\begin{align}
		& \frac{1}{d}(X_{\mA}^\top X_{\mA})E_{\mA} - \frac{1}{d}X_{\mA}^\top W+\lambda_d\tilde{S}_{\mA} = 0\label{eqA1}\\
		& \frac{1}{d}(X_{\mA^c}^\top X_{\mA})E_{\mA} - \frac{1}{d}X_{\mA^c}^\top W+\lambda_d\tilde{S}_{\mA^c} = 0\label{eqA2}
	\end{align}
	\end{subequations}
	Solving~\eqref{eqA1} with respect to $E_{\mA}$ and substituting the solution in~\eqref{eqA2} completes the proof.
\end{proof}

\subsection{Proof of Lemma~\ref{sparsity}:} As shown in Lemma~\ref{pdw}, it is enough to prove that the PDW succeeds with high probability.
To this goal, we show that $\max_{i\in\mA^c}\|\tilde{S}^{(i)}\|_1<1$ with high probability, which results in the success of the PDW approach.
 Lemma~\ref{lerror} yields that
\begin{align}
\|\tilde{S}^{(i)}\|_1 \leq& \underbrace{\left\|\frac{1}{d\lambda_d}\left({X_{(i)}}^\top -({X_{(i)}}^\top X_{\mA})(X_{\mA}^\top X_{\mA})^{-1}X_{\mA}^\top\right)W\right\|_1}_{Z_1^{(i)}}+\underbrace{\left\|\frac{1}{d}{X_{(i)}}^\top X_{\mA}(\frac{1}{d}X_{\mA}^\top X_{\mA})^{-1}\tilde{S}_{\mA}\right\|_1}_{Z_2^{(i)}}
\end{align}
Similar to~\cite{Martin11}, we will show that $\max_{i\in\mA^c}Z_1^{(i)}<\gamma/2$ and $\max_{i\in\mA^c}Z_2^{(i)}<1-\gamma/2$ with high probability.
First, consider $\max_{i\in\mA^c}Z_1^{(i)}$. We have
\begin{equation}
Z_1^{(i)} \!=\! \sum_{(k,l)\in\Theta^{(i)}}\Big|\underbrace{\frac{1}{d\lambda_d}{(X_{(i)})_{:,k}}^\top(I-X_{\mA}(X_{\mA}^\top X_{\mA})^{-1}X_{\mA}^\top)W_{:,l}}_{R^{(i)}_{kl}}\Big|
\end{equation}
Given $X$, note that $R_{kl}^{(i)}$ is Gaussian with variance 
\begin{equation}
\frac{\sigma_w^2}{d^2\lambda_d^2}\left({(X_{(i)})_{:,k}}^\top(I-X_{\mA}(X_{\mA}^\top X_{\mA})^{-1}X_{\mA}^\top)^2{(X_{(i)})_{:,k}}\right)
\end{equation}
Moreover, $X_{\mA}(X_{\mA}^\top X_{\mA})^{-1}X_{\mA}^\top$ is an orthogonal projection onto the range of $X_{\mA}$. Therefore,
\begin{align}
&\frac{\sigma_w^2}{d^2\lambda_d^2}\left({(X_{(i)})_{:,k}}^\top(I-X_{\mA}(X_{\mA}^\top X_{\mA})^{-1}X_{\mA}^\top)^2{(X_{(i)})_{:,k}}\right)\nonumber\\
& = \frac{\sigma_w^2}{d^2\lambda_d^2}\left({(X_{(i)})_{:,k}}^\top(I-X_{\mA}(X_{\mA}^\top X_{\mA})^{-1}X_{\mA}^\top){(X_{(i)})_{:,k}}\right)\nonumber\\
& \leq \frac{\sigma_w^2}{d^2\lambda_d^2}\|(X_{(i)})_{:,k}\|_2^2
\end{align}
Define $p_i = n_i$ if $1\leq i\leq \bar{n}$ and $p_i = m_i$ if $\bar{n}+1\leq i\leq\bar{n}+\bar{m}$. Due to Lemma~\ref{l2}, the last inequality is upper bounded by ${4\sigma_w^2\sigma_{\max}^2}/{d\lambda_d^2}$ for every $k\in\{1,...,p_i\}$ with probability of at least $1-\exp(-d/2+\log p_i)$ for $d>2\log p_i$. Conditioned on this event, one can write:
\begin{equation}
Z_1^{(i)} = \max_{\epsilon\in\{-1,+1\}^{p_i\times n_j}}\sum_{(k,l)\in\Theta^{(i)}} \epsilon_{kl} R_{kl}^{(i)}
\end{equation}
which means that $\sum_{(k,l)\in\Theta^{(i)}} \epsilon_{kl} R_{kl}^{(i)}$ is sub-Gaussian with the parameter ${4D_j\sigma_w^2\sigma_{\max}^2}/{d\lambda_d^2}$. This implies that
\begin{align}
\mathbb{P}(\max_{i\in\mA^c}Z_1^{(i)}\!\geq\! \zeta) &\!=\! \mathbb{P}\left(\max_{i\in\mA^c}\max_{\epsilon\in\{-1,+1\}^{p_i\times n_j}}\sum_{(k,l)\in\Theta^{(i)}} \epsilon_{kl} R_{kl}^{(i)}\geq \zeta\right)\nonumber\\
&\leq \!2\exp\left(-\frac{d\lambda_d^2\zeta^2}{8D_j\sigma_w^2\sigma_{\max}^2}\!+\!D_j\!+\!\log(\bar{n}\!+\!\bar{m})\right)+\exp\left(-d/2+\log p_{\max}+\log (\bar{n}+\bar{m})\right)
\end{align}
where we have used Lemma~\ref{l1}, the second statement of Lemma~\ref{l4} and the facts that $p_i\leq p_{\max}$ and $|\mA^c|\leq \bar{n}+\bar{m}$ in the last inequality. Now, setting $\zeta = \gamma/2$ and
\begin{equation}
\lambda_d\geq \sqrt{\frac{32c_1\sigma_w^2\sigma_{\max}^2}{\gamma^2}\cdot\frac{(D_j)^2+D_j\log(\bar{n}+\bar{m})}{d}}
\end{equation}
for some arbitrary constant $c_1>1$ yields that
\begin{align}
\mathbb{P}(\max_{i\in\mA^c}Z_1^{(i)}< \gamma/2)\geq& 1-2\exp(-(c_1-1)(D_j+\log (\bar{n}+\bar{m})))- \exp(-d/2+\log p_{\max}+\log (\bar{n}+\bar{m}))\nonumber\\
\geq& 1-3\exp(-(c_1-1)(D_j+\log (\bar{n}+\bar{m})))
\end{align}
where the last inequality is due to the lower bound~\eqref{d3} on $d$. 
Next, an upper bound on $\max_{i\in\mathcal{A}^c}Z_2^{(i)}$ will be derived. Since each row of $X$ is drawn from $N(0,\tilde\Sigma)$, one can write the distribution of $X_{\mA^c}^\top$, conditioned on $X_{\mA}$ as
\begin{equation}\label{eq562}
N\Big(\tilde\Sigma_{\mathcal{A}^c,\mathcal{A}}(\tilde\Sigma_{\mathcal{A},\mathcal{A}})^{-1}X_{\mA}^\top,\underbrace{\tilde\Sigma_{\mathcal{A}^c,\mathcal{A}^c}-\tilde\Sigma_{\mathcal{A}^c,\mathcal{A}}(\tilde\Sigma_{\mathcal{A},\mathcal{A}})^{-1}\tilde\Sigma_{\mathcal{A},\mathcal{A}^c}}_{\tilde{\Sigma}_{\mA^c|\mA}}\Big)
\end{equation}
Based on~\eqref{eq562}, one can verify that $\frac{1}{d}X_{\mA^c}^\top X_{\mA}(\frac{1}{d}X_{\mA}^\top X_{\mA})^{-1}\tilde{S}_{\mA}$ has the same distribution as 
\begin{equation}
\tilde\Sigma_{\mathcal{A}^c,\mathcal{A}}(\Sigma_{\mathcal{A},\mathcal{A}})^{-1}\tilde{S}_{\mA}+\frac{1}{d}V^\top X_{\mA}(\frac{1}{d}X_{\mA}^\top X_{\mA})^{-1}\tilde{S}_{\mA}
\end{equation}
where $V$ is a random matrix with zero mean, covariance $\tilde{\Sigma}_{\mA^c|\mA}$, and independent of $X$. In light of the definition of $\tilde{\Sigma}_{\mA^c|\mA}$, it can be easily seen that the elements of $V$ are sub-Gaussian with parameters of at most $\sigma_{\max}^2$. This implies that
\begin{align}\label{W}
\!\!\!\max_{i\in\mA^c} Z_2^{(i)}\leq& \max_{i\in\mA^c}\left\|\Sigma_{i,\mathcal{A}}(\Sigma_{\mathcal{A},\mathcal{A}})^{-1}\tilde{S}_{\mA}\right\|_1+\max_{i\in\mA^c}\left\|\frac{1}{d}{V_{(i)}}^\top X_{\mA}(\frac{1}{d}X_{\mA}^\top X_{\mA})^{-1}\tilde{S}_{\mA}\right\|_1\nonumber\\
\leq& 1-\gamma+\max_{i\in\mA^c}\underbrace{\left\|\frac{1}{d}{V_{(i)}}^\top X_{\mA}(\frac{1}{d}X_{\mA}^\top X_{\mA})^{-1}\tilde{S}_{\mA}\right\|_1}_{Z_3^{(i)}}
\end{align}
where we have used the mutual incoherence property and the fact that $\|\tilde{S}^{(i)}\|_1 = 1$ for every $i\in\mA$. Now, it remains to show that $\max_{i\in\mA^c}Z_3^{(i)}<\gamma/2$ with high probability. Similar to $Z_1^{(i)}$, one can write:
\begin{equation}
	Z_3^{(i)} = \sum_{(k,l)\in\Theta^{(i)}}\Big|\underbrace{\frac{1}{d}{(V_{(i)})_{:,k}}^\top X_{\mA}(\frac{1}{d}X_{\mA}^\top X_{\mA})^{-1}(\tilde{S}_{\mA})_{:,l}}_{T^{(i)}_{kl}}\Big|
\end{equation}
Given $X$, note that $T^{(i)}_{kl}$ is Gaussian with variance
\begin{equation}\label{Tvar}
	\sigma_{\max}^2(\tilde{S}_{\mA})_{:,l}^\top\left(\frac{1}{d}X_{\mA}^\top X_{\mA}\right)^{-1}(\tilde{S}_{\mA})_{:,l}
\end{equation}
Also, $\|(\tilde{S}_{\mA})_{:,l}\|_2^2\leq k_j$. Therefore, Lemma~\ref{l3} can be used to bound~\eqref{Tvar} as follows:
\begin{align}\label{eq56}
(\tilde{S}_{\mA})_{:,l}^\top\left(\frac{1}{d}X_{\mA}^\top X_{\mA}\right)^{-1}(\tilde{S}_{\mA})_{:,l}&\leq \sigma_{\max}^2k_j\left\|\frac{1}{d}(\frac{1}{d}X_{\mA}^\top X_{\mA})^{-1}\right\|_2\nonumber\\
&\leq \sigma_{\max}^2k_j\left(\frac{1}{d}\frac{8}{\Lambda_{\min}}+\frac{1}{d}\left\|\Sigma_{\mathcal{A}, \mathcal{A}}^{-1}\right\|_2\right)\nonumber\\
&\leq \sigma_{\max}^2k_j\left(\frac{1}{d}\cdot\frac{8}{\Lambda_{\min}}+\frac{1}{d}\cdot\frac{1}{\Lambda_{\min}}\right)\nonumber\\
&\leq \frac{9\sigma_{\max}^2k_j}{\Lambda_{\min}d}
\end{align}
with probability of at least $1-2\exp(-d/{2})$. Similar to the arguments made for bounding $\max_{i\in\mA^c}Z_1^{(i)}$, one can verify that
\begin{align}
\mathbb{P}\left(\max_{i\in\mA^c}Z_3^{(i)}\!\!<\!\! \gamma/2\right)\!\geq 1\!-\!2&\exp\Big(-\frac{\Lambda_{\min}d\gamma^2}{72\sigma_{\max}^2 k_jD_j}\!+\!D_j+\!\log(\bar{n}+\bar{m})\Big)- 2\exp\Big(-\frac{d}{2}\Big)
\end{align}
Now, choosing 
\begin{equation}
d\geq \frac{72c_2\sigma_{\max}^2 k_{j}D_j}{\Lambda_{\min}\gamma^2}\cdot(D_j+\log(\bar{n}+\bar{m}))
\end{equation}
for some arbitrary constant $c_2>1$ results in
\begin{align}
\mathbb{P}\left(\max_{i\in\mA^c}Z_3^{(i)}\!<\! \gamma/2\right)\!\geq& 1\!-\!4\exp(-(c_2-1)(D_j\!+\!\log(\bar{n}+\bar{m})))
\end{align}
Therefore, $\max_{i\in\mA^c}\|\tilde{S}^{(i)}\|_1<1$ and, hence, PDW succeeds with a probability that is lower bounded by~\eqref{p3}.\qed

\subsection{Proof of Lemma~\ref{errorbound}:}

In order to bound the estimation error, an upper bound on $\|E\|_\infty$ will be derived, conditioning on the success of the PDW approach. Note that $E_{\mA^c} = 0$ according to Lemma~\ref{lerror} and, hence, it suffices to bound $\|E_{\mA}\|_\infty$. Again, due to Lemma~\ref{lerror}, one can write:
\begin{align}
\!\!\!\!\max_{k = 1,...,n_j}\|(E_{\mA})_{:,k}\|_\infty&\!\leq\!\! \max_{k = 1,...,n_j}\!\!\underbrace{\left\|(\frac{1}{d}X_{\mA}^\top X_{\mA})^{-1}\frac{1}{d}X_{\mA}^\top W_{:,k}\right\|_\infty}_{Z_4^k}+\!\max_{k = 1,...,n_j}\underbrace{\left\|(\frac{1}{d}X_{\mA}^\top X_{\mA})^{-1}\lambda_d(\tilde{S}_{\mA})_{:,k}\right\|_\infty}_{Z_5^k}
\end{align}
for $k = 1,2,...,n_j$. For bounding $Z_5^k$, it can be argued similarly to~\eqref{eq56} that
\begin{align}\label{eq61}
\max_{k = 1,...,n_j}Z_5^k\leq&\max_{k = 1,...,n_j} \left\|\left((\frac{1}{d}X_{\mA}^\top X_{\mA})^{-1}-\Sigma_{\mathcal{A}, \mathcal{A}}^{-1}\right)\lambda_d(\tilde{S}_{\mA})_{:,k}\right\|_\infty+\max_{k = 1,...,n_j}\left\|\Sigma_{\mathcal{A}, \mathcal{A}}^{-1}\lambda_d(\tilde{S}_{\mA})_{:,k}\right\|_\infty\nonumber\\
\leq& \left\|(\frac{1}{d}X_{\mA}^\top X_{\mA})^{-1}-\Sigma_{\mathcal{A}, \mathcal{A}}^{-1}\right\|_2\lambda_d\sqrt{k_j}+\frac{\lambda_d}{\Lambda_{\min}}\nonumber\\
\leq& \frac{\lambda_d}{\Lambda_{\min}}\left({8\sqrt{k_j}\sqrt{\frac{k_jn_j+c_3\log(\bar{n}+\bar{m})}{d}}}+1\right)
\end{align}
for some $c_3> 0$ with probability of at least $1-2\exp(-({k_jn_j+c_3\log(\bar{n}+\bar{m})})/{2})$, where we have used the matrix norm properties and Lemma~\ref{l3} with $t = {k_jn_j+c_3\log(\bar{n}+\bar{m})}$ (note that $|I(\mA)|\leq k_jn_j$). Now, it remains to bound $\max_{k = 1,...,n_j}Z_4^k$. This can be carried out similar to the previous arguments, i.e., by making use of~\eqref{eq56} and obtaining a sub-Gaussian parameter for $(\frac{1}{d}X_{\mA}^\top X_{\mA})^{-1}\frac{1}{d}X_{\mA}^\top W_{:,k}$. For brevity, only the final key relation is stated below:
\begin{align}
\mathbb{P}(\max_{k = 1,...,n_j}\!Z_4^k\!\geq\! \zeta)\leq& 2\exp\Big(\!-\!\frac{d\Lambda_{\min}\zeta^2}{18\sigma_w^2}\!+\!\log n_j+\!\log(k_jn_j)\Big)+2\exp(\!-\!{d}/{2})\nonumber\\
\leq& 2\exp\Big(\!-\!\frac{d\Lambda_{\min}\zeta^2}{18\sigma_w^2}\!+\!2(\alpha_n\!+\!\alpha_k)\!\log(\bar{n}\!+\!\bar{m})\Big)+2\exp(\!-\!{d}/{2})
\end{align}
where the last inequality is due to the assumption that $n_j\leq n_{\max} = O\left((\bar{n}+\bar{m})^{\alpha_n}\right)$ and $k_j\leq k_{\max} = O\left((\bar{n}+\bar{m})^{\alpha_k}\right)$. Now, setting 
\begin{equation}
\zeta = \sqrt{\frac{36c_4(\alpha_n+\alpha_k)\sigma_w^2\log(\bar{n}+\bar{m})}{d\Lambda_{\min}}}
\end{equation}
for an arbitrary constant $c_4>1$, together with the inequality $\log r_1\leq \log(k_jD_j)$, leads to
\begin{align}
\max_{k = 1,...,n_j}\!\!Z_4^k\!\leq\! \sqrt{\frac{36c_4(\alpha_n+\alpha_k)\sigma_w^2\log(\bar{n}+\bar{m})}{d\Lambda_{\min}}}
\end{align}
with probability of at least
\begin{equation}
1-2\exp\left(-2(c_4-1)(\alpha_n+\alpha_k)\log(\bar{n}+\bar{m})\right)-2\exp(-{d}/{2})
\end{equation}
Combining this inequality with~\eqref{eq61} results in the elementwise error bound~\eqref{element_error} with probability of at least~\eqref{p4}. This concludes the proof.\qed

\subsection{Proof of Theorem~\ref{thm2}:}

First, we present the sketch of the proof in a few steps:
\begin{itemize}
	\item[1.] We decompose the block-regularized problem \eqref{lasso2} into $\bar{n}$ disjoint block-regularized subproblems defined in~\eqref{lasso3}.
	\item[2.] For each of these subproblems, we consider the event that Lemmas~\ref{sparsity} and~\ref{errorbound} hold. 
	\item[3.] We consider the intersection of these $\bar{n}$ events and show that, together with~\eqref{lambda2} and~\eqref{d22}, they lead to the element-wise error~\eqref{error2} with probability of at least $1-\delta$. 
\end{itemize}

\textit{Step 1:} \eqref{lasso2} can be rewritten as follows:
\begin{equation}\label{lasso4}
\hat{\Theta} = \arg\min_{\Theta}\sum_{j=1}^{n}\left(\frac{1}{2d}\|Y^{(:,j)}-X\Theta^{(:,j)}\|_2^2 + \lambda\|\Theta^{(:,j)}\|_\bl\right)
\end{equation}
The above optimization problem can be naturally decomposed into $\bar{n}$ disjoint block-regularized subproblems in the form of~\eqref{lasso3}.

\textit{Step 2:} Assume that~\eqref{d3} and~\eqref{lambda3} hold
for every $1\leq j\leq \bar{n}$. Upon defining $\mathcal{T}_j$ as the event that  Lemmas~\ref{sparsity} and~\ref{errorbound} hold, one can write:
\begin{equation}\label{p33}
\begin{aligned}
\mathbb{P}(\mathcal{T}_j)\geq 1&-5\exp\big(-(c_1-1)(D_j+\log(\bar{n}+\bar{m}))\big)\\
&-4\exp\big(-(c_2-1)(D_j+\log(\bar{n}+\bar{m}))\big)\\
&-2\exp(-({k_jn_j+c_3\log(\bar{n}+\bar{m})})/{2})\\
&-2\exp\big(-2(c_4-1)(\alpha_n+\alpha_k)\log(\bar{n}+\bar{m}))\big)
\end{aligned}
\end{equation}
For every $1\leq j\leq \bar{n}$.

\textit{Step 3:} Assume that $c_1, c_2, c_4>2$ and $c_3>1$. Consider the event $\mathcal{T} = \mathcal{T}_1\cap\mathcal{T}_2\cap\cdots\cap\mathcal{T}_n$. Based on~\eqref{p33} and Lemma~\ref{l4}, one can write:
\begin{align}\label{probsum2}
\mathbb{P}(\mathcal{T})\!\geq\!
1&\!-\!\underbrace{K_1(\bar{n}+\bar{m})^{-(c_1-2)}}_{(a)}\!-\!\underbrace{K_2(\bar{n}+\bar{m})^{-(c_2-2)}}_{(b)}\!-\!\underbrace{K_3(\bar{n}+\bar{m})^{-(\frac{c_3}{2}-1)}}_{(c)}\!-\!\underbrace{K_4(\bar{n}+\bar{m})^{-(2(\alpha_n+\alpha_k)(c_4-1)-1)}}_{(d)}
\end{align}
for some constants $K_1, K_2, K_3, K_4$. One can easily verify that the following equalities are enough to guarantee that the right hand side of~\eqref{probsum2} is equal to $1-\delta$:
\begin{align}\label{eqc}
	& c_1 = \frac{\log(4K_1/\delta)}{\log(\bar{n}+\bar{m})}+2,\nonumber\\
	& c_2 = \frac{\log(4K_2/\delta)}{\log(\bar{n}+\bar{m})}+2,\nonumber\\
	& c_3 = c_1 = \frac{2\log(4K_3/\delta)}{\log(\bar{n}+\bar{m})}+2,\nonumber\\
	&c_4 = \frac{\log(4K_4/\delta)}{2(\alpha_n+\alpha_k)\log(\bar{n}+\bar{m})}+\frac{1}{2(\alpha_n+\alpha_k)}+1.
\end{align}
Substituting~\eqref{eqc} in Lemmas~\ref{sparsity} and~\ref{errorbound} leads to two observations:
\begin{itemize}
	\item[-] If $\lambda_d$ and $d$ satisfy~\eqref{lambda2} and~\eqref{d22}, then they also satisfy~\eqref{lambda3} and~\eqref{d3}. 
	\item[-] The parameter $g$ defined in~\eqref{error2} is greater than or equal to $g_j$ for every $j = 1,...,\bar{n}$.
\end{itemize} 
Therefore,~\eqref{lambda2} and~\eqref{d22} guarantee that: 1) $\hat{\Theta}$ is unique and does not have any false positive in its blocks, and 2) its element-wise error is upper bounded by~\eqref{error2}. Now, it only remains to show that $\hat{\Theta}$ excludes false negatives (the blocks that are mistakenly estimated to have nonzero values). To this goal, it suffices to show that~\eqref{d22} guarantees $g<t_{\min}$. Suppose that 
\begin{equation}
	d = \Omega\left(C_\Theta\kappa(\tilde{\Sigma})^2k_{\max}\left(D\log(\bar{n}+\bar{m})+D^2\log(1/\delta)\right)\right)
\end{equation}
In what follows, we will show that $C_\Theta = O(1)$ is enough to have $g<t_{\min}$. The lower bound on $d$ in \eqref{d22} yields that
\begin{equation}
	g\leq K\left(\frac{1}{\sqrt{C_{\Theta}k_{\max}}}+\frac{1}{C_{\Theta}\kappa(\tilde{\Sigma})}\right)
\end{equation}
for some constant $K$. Therefore,
\begin{equation}
	C_{\Theta} =  \frac{2/K}{t_{\min}\kappa(\tilde{\Sigma})}+\frac{4/K}{t_{\min}^2k_{\max}} = O(1)
\end{equation}
is enough to ensure $g<t_{\min}$. This completes the proof. \qed

\subsection{Proof of Theorem~\ref{thm1}:}

The proof immediately follows from Theorem~\ref{thm2} by noting that $D = n_{\max} = 1$, $\bar{n} = n$, and $\bar{m} = m$.\qed

\section{Numerical Results}\label{sec:num}

In this section, we illustrate the performance of the block-regularized estimator and compare it with its least-squares counterpart. We consider three case studies on synthetically generated systems, physical mass-spring networks, and multi-agent systems. 
The simulations are run on a laptop computer with an Intel Core i7 quad-core 2.50 GHz CPU and 16GB RAM. The reported results are for a serial implementation in MATLAB R2017b, and the function \texttt{lasso} is used to solve~\eqref{lasso2}. Define the (block) mismatch error as the total number of false positives and false negatives in the (block) sparsity pattern of the estimator. Moreover, define \textit{relative number of sample trajectories} (RST) as the number of sample trajectories normalized by the dimension of the system, and \textit{relative (block) mismatch error} (RME) as the mismatch error normalized by total number of elements (blocks) in $\Theta$.
To verify the developed theoretical results, $\lambda_d$ is set to
\begin{equation}\label{lambda33}
\sqrt{\frac{2(D^2+D\log(\bar{n}+\bar{m}))}{d}}
\end{equation}
in all of the experiments. Note that this choice of $\lambda_d$ does not require any additional fine-tuning.

\subsection{Case Study 1: Synthetically Generated Systems}
Given the numbers $n$ and $w$, and for each instance of the problem, the state and input matrices are constructed as follows: The size of each block in $A$ and $B$ is set to 1. The diagonal elements of $A\in\mathbb{R}^{n\times n}$ and $B\in\mathbb{R}^{n\times n}$ are set to 1 (the dimensions of the inputs and states are chosen to be equal). The elements of the first $w$ upper and lower diagonals of $A$ and $B$ are set to $0.3$ or $-0.3$ with equal probability. Furthermore, at each row of $A$, another $w$ elements are randomly chosen from the elements not belonging to the first $w$ upper and lower diagonals and set to $0.3$ or $-0.3$ with equal probability. The structure of $[A\ \ B]$ is visualized in Figure~\ref{fig_sp} for $w = 2,3,4$. Based on this procedure, the number of nonzero elements at each row of $[A\ \ B]$ varies between $3w+2$ and $5w+2$. We set $\Sigma_u = I$ and $\Sigma_w = 0.5I$.
The mutual incoherence property is satisfied for all constructed instances.

In the first set of experiments, we consider the mismatch error of $\hat\Theta$ with respect to the number of sample trajectories and for different system dimensions. The length of the time horizon $T$ is set to $3$. The results are illustrated in Figure~\ref{fig_d} for $n+m$ equal to $200$, $600$, $1200$, and $2000$. In all of these test cases, $w$ is chosen in such a way that the number of nonzero elements in each column of $\Theta$ is between $(n+m)^{0.3}$ and $(n+m)^{0.4}$. It can be observed that as the dimension of the system increases, a higher number of sample trajectories is required to have a small mismatch error in the block-regularized estimator. Conversely, the required value of RST to achieve a small RME reduces as the dimension of the system grows. More precisely, RST should be at least $1.80$, $1.13$, $0.37$, and $0.20$ to guarantee $\text{RME}\leq 0.1\%$, when $m+n$ is equal to $200$, $600$, $1200$, and $2000$, respectively. 

In the next set of experiments, we consider the mismatch error for different time horizons $T = 3,4,...,7$, by fixing $m+n = 600$ and $w = 2$. As mentioned before, large values of $T$ tend to inflate the easily identifiable modes of the system and suppress the nearly hidden ones, thereby making it hard to obtain an accurate estimation of the parameters. It is pointed out that $\kappa(F_TF_T^\top+G_TG_T^\top)$ is a good indicator of the gap between these modes. This relationship is clearly reflected in Figures~\ref{fig_T} and~\ref{fig_T_cond}. As can be observed in Figure~\ref{fig_T}, $330$ sample trajectories are enough to guarantee $\text{RME}\leq 0.1\%$ for $T=3$. However, for $T=7$, RME cannot be reduced below $0.42\%$ even with 1000 sample trajectories. To further elaborate on this dependency, Figure~\ref{fig_T_cond} is used to illustrate the value of $\kappa(F_TF_T^\top+G_TG_T^\top)$ with respect to $T$ in a $\log$-$\log$ scale. One can easily verify that $\kappa(F_TF_T^\top+G_TG_T^\top)$ associated with $T = 7$ is $485$ times greater than this parameter for $T=3$.

Finally, we study the block-regularized estimator for different per-column numbers of nonzero elements in $\Theta$ and compare its accuracy to the least-squares estimator. Fixing $T = 3$ and $m+n = 600$, Figure~\ref{fig_w} depicts the mismatch error of the block-regularized estimator when the maximum number of nonzero elements at each column of $\Theta$ ranges from $7$ (corresponding to $w=1$) to $27$ (corresponding to $w=5$). Not surprisingly, the required number of samples to achieve a small mismatch error increases as the number of nonzero elements in each column of $\Theta$ grows. On the other hand, the least-squares estimator is fully dense in all of these experiments, regardless of the number of sample trajectories. To have a better comparison between the two estimators, we consider the $2$-norm of the estimation errors normalized by the $2$-norm of $\Theta^*$, for different numbers of nonzero elements in each column of $\Theta^*$. As it is evident in Figure~\ref{fig_w_error}, the block-regularized estimator significantly outperforms the least-squares one for any number of sample trajectories. Furthermore, the least-squares estimator is not defined for $d<600$.

\begin{figure*}
	\centering
	\subfloat[$w = 2$]{\label{fig_w2}
		\includegraphics[width=.3\columnwidth]{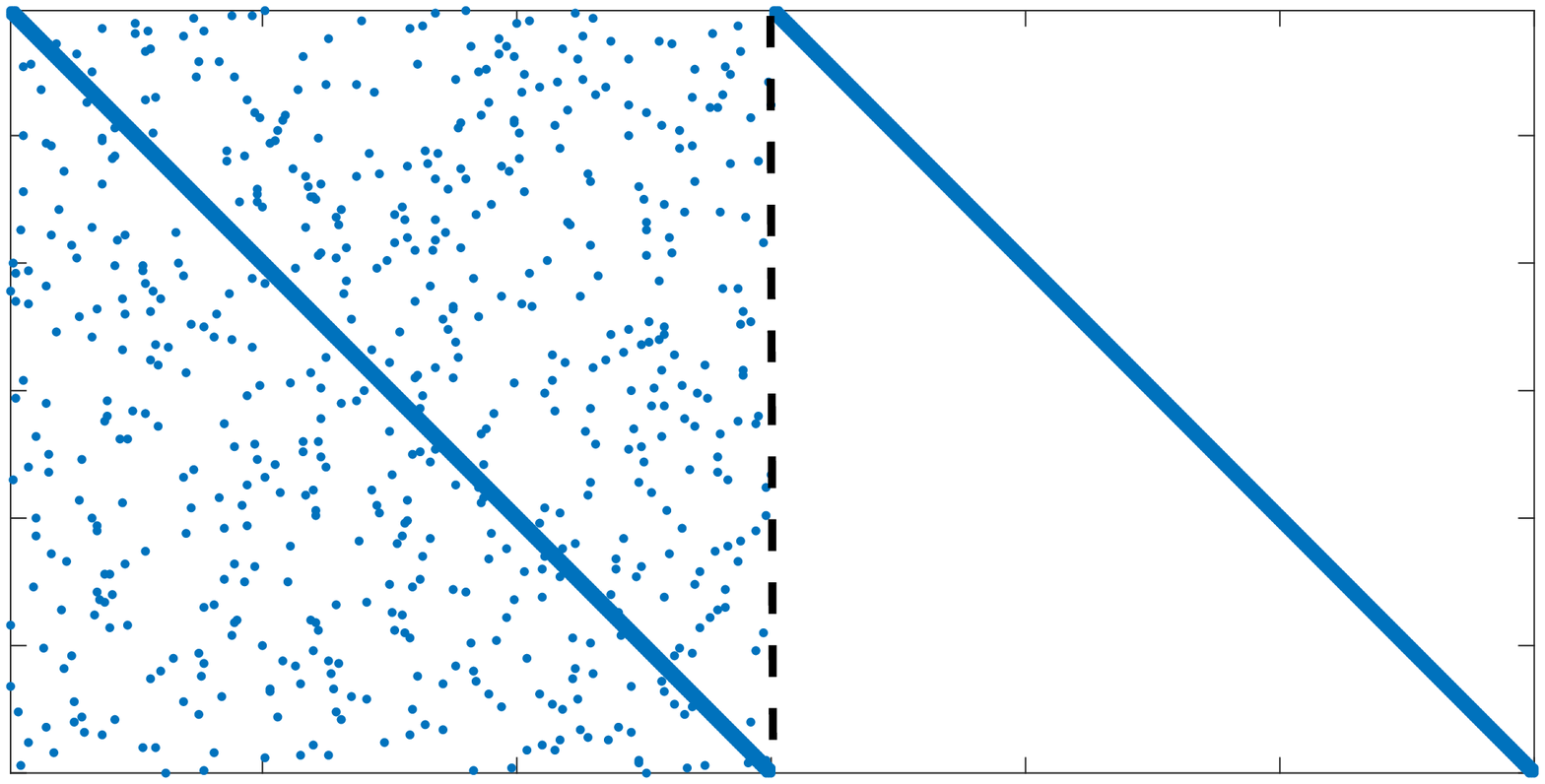}}
	\subfloat[$w = 3$]{\label{fig_w3}
		\includegraphics[width=.3\columnwidth]{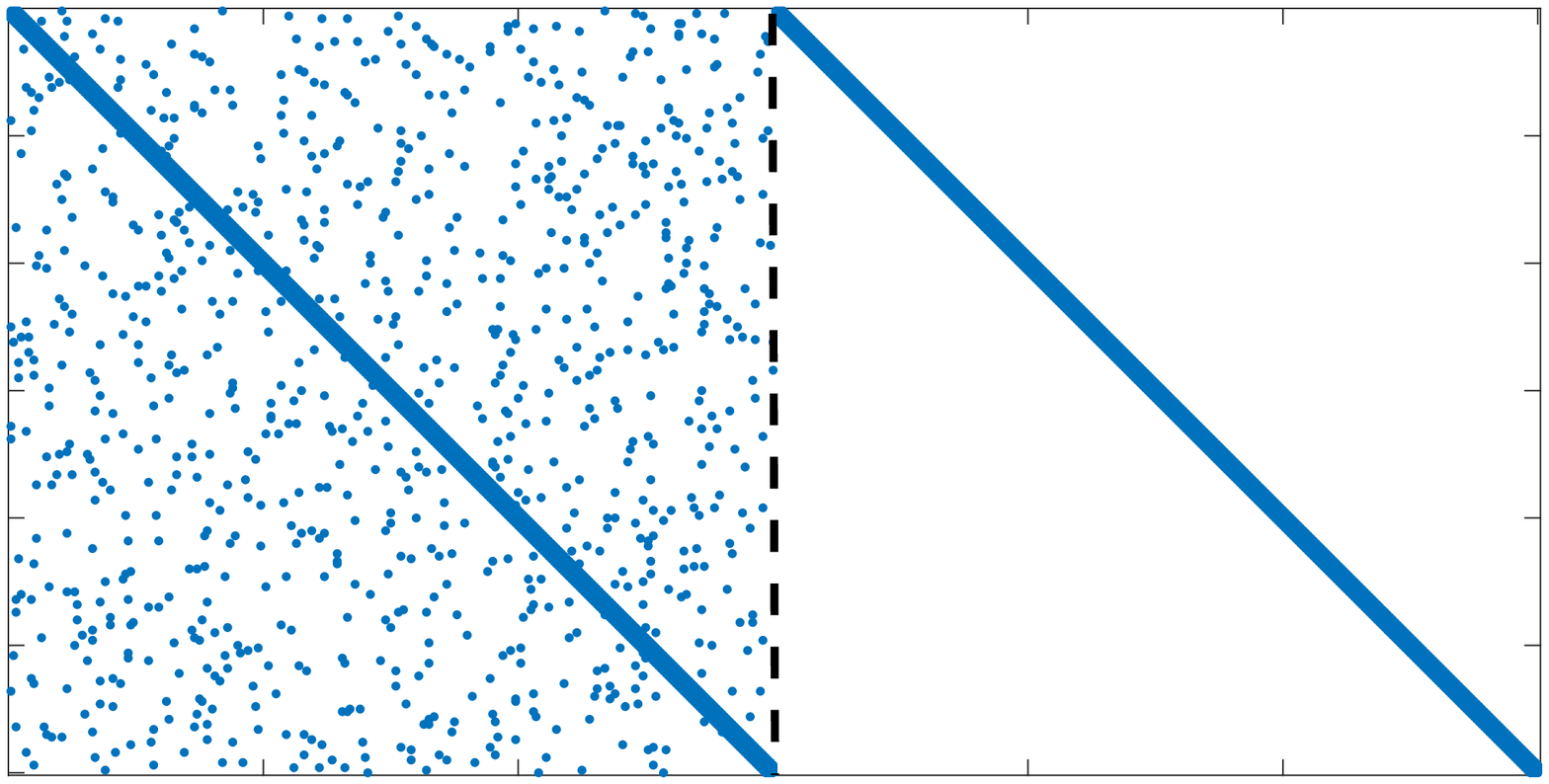}}
	\subfloat[$w = 4$]{\label{fig_w4}
		\includegraphics[width=.3\columnwidth]{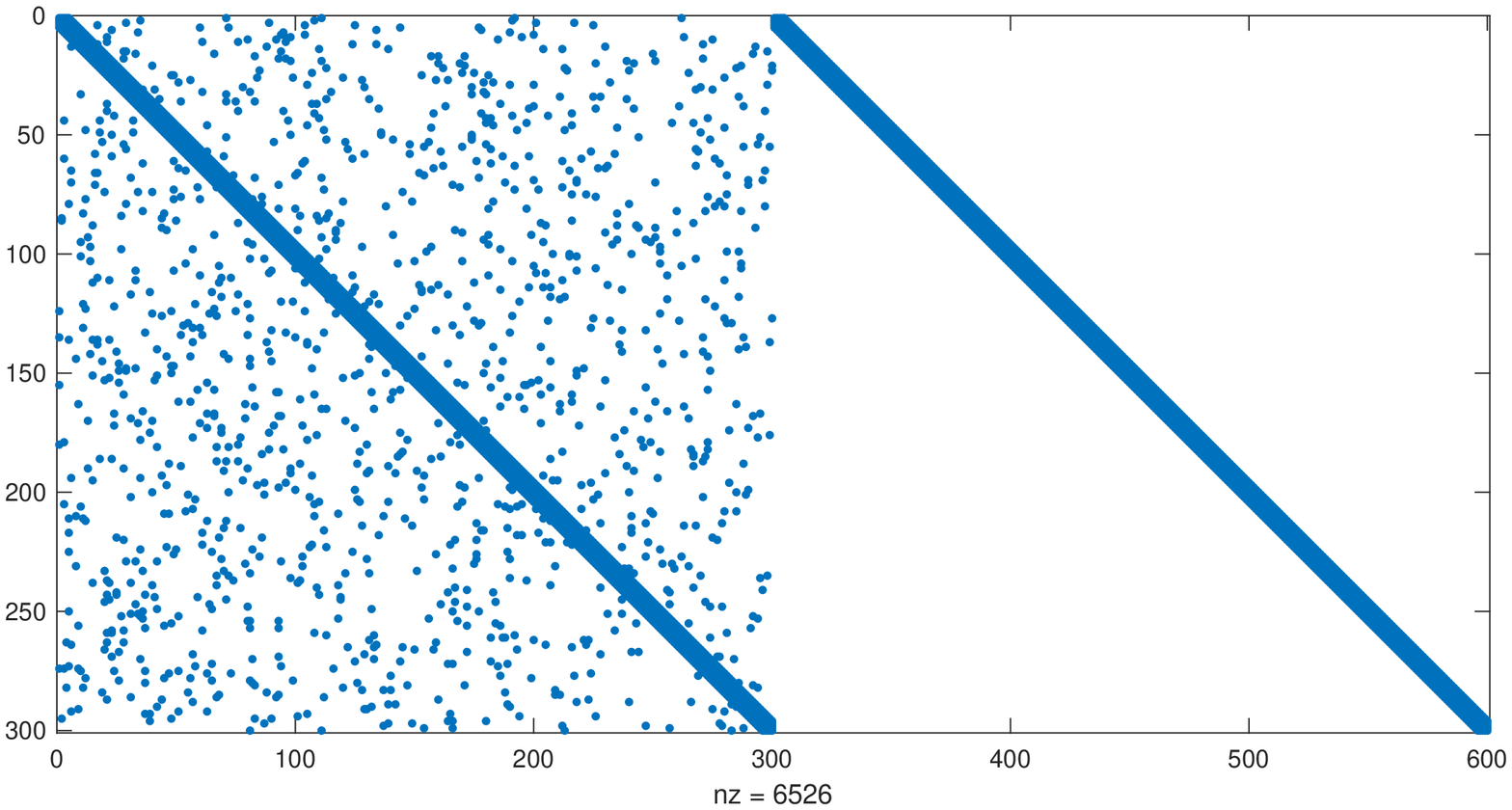}}
	\caption{ \footnotesize The sparsity structure of the matrix $[A\ \ B]$ for $w = 2,3,4$.}\label{fig_sp}
\end{figure*}

\begin{figure*}
	\centering
	\subfloat[Mismatch error for different $d$]{\label{fig_d}
		\includegraphics[width=.32\columnwidth]{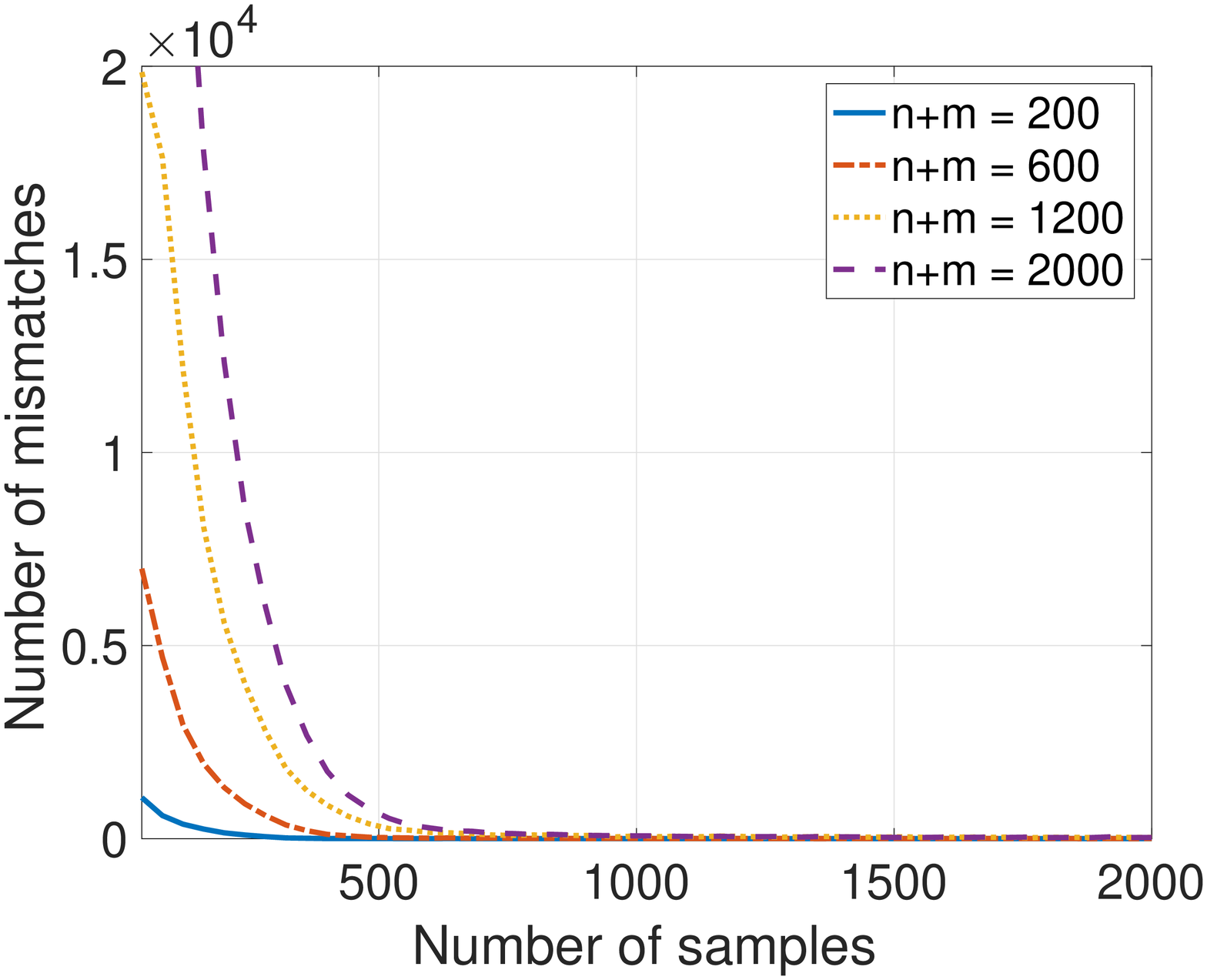}}
	\subfloat[Mismatch error for different $T$]{\label{fig_T}
		\includegraphics[width=.32\columnwidth]{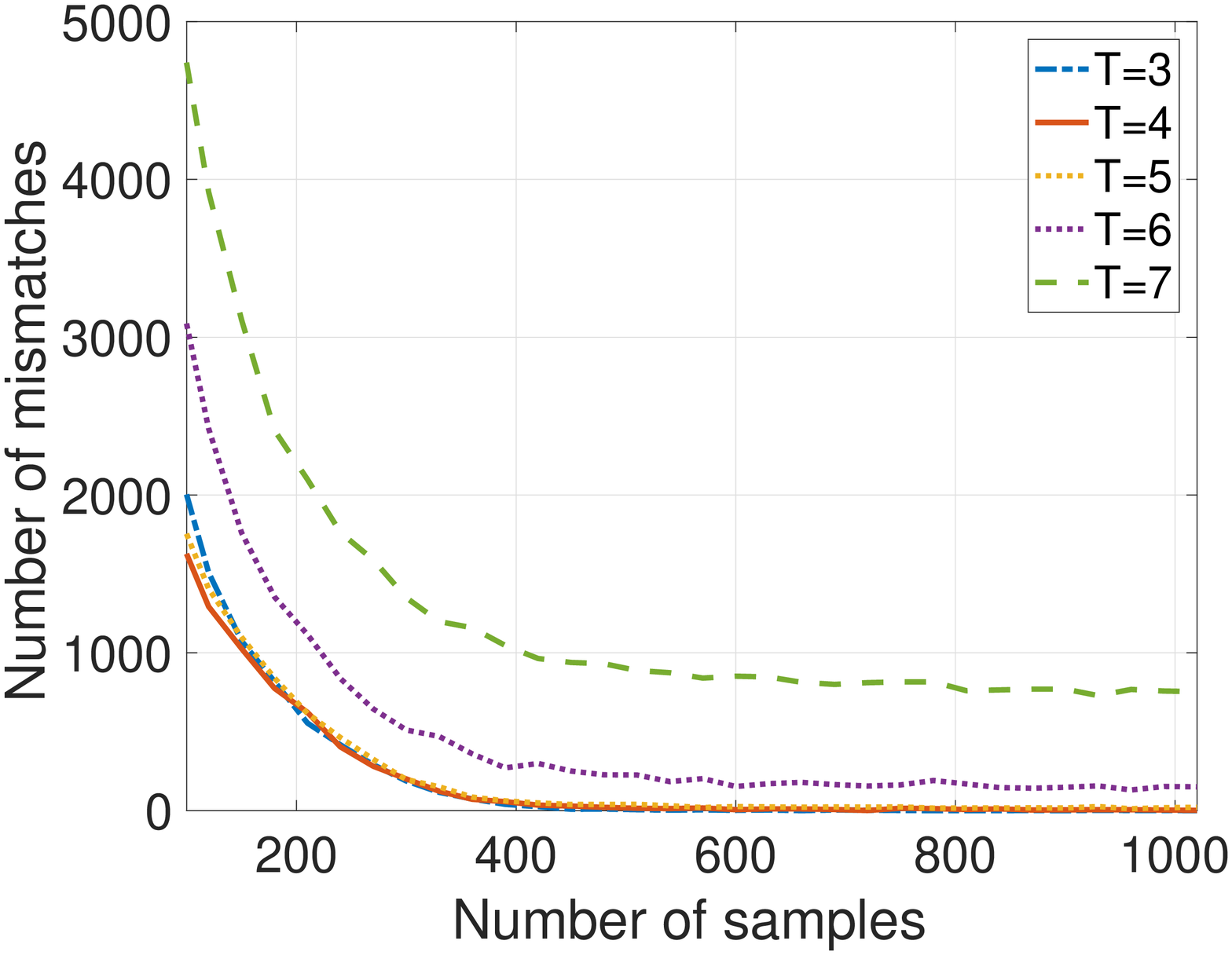}}
	\subfloat[Condition number for different $T$]{\label{fig_T_cond}
		\includegraphics[width=.32\columnwidth]{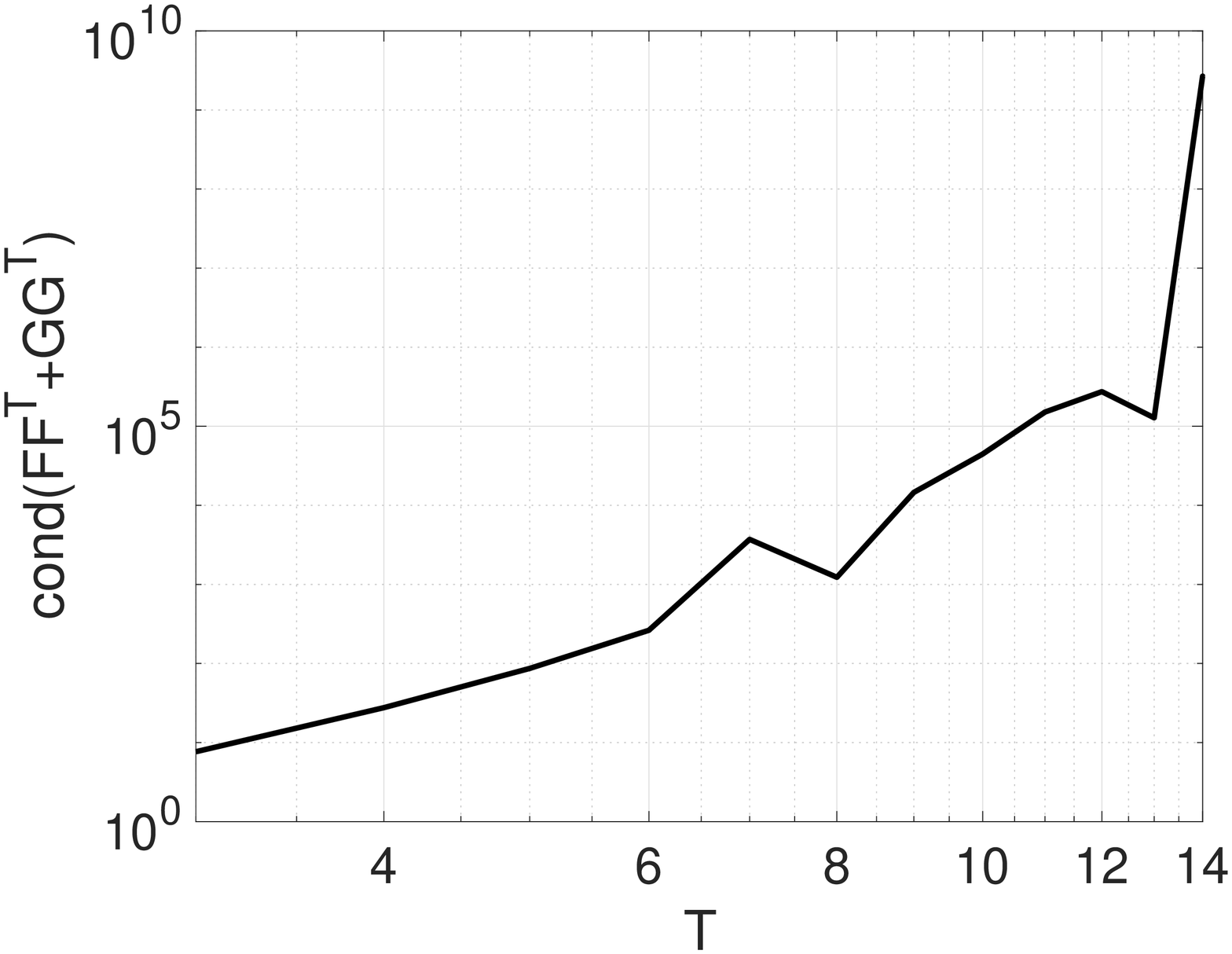}}
	\caption{ \footnotesize (a) The mismatch error with respect to the number of sample trajectories for different system dimensions, (b) the mismatch error with respect to the number of sample trajectories for different time horizons, (c) the condition number of $FF^\top+GG^\top$ with respect to the time horizon.}
\end{figure*}

\begin{figure*}
	\centering
	\subfloat[Mismatch error for different $k_{\max}$]{\label{fig_w}
		\includegraphics[width=.32\columnwidth]{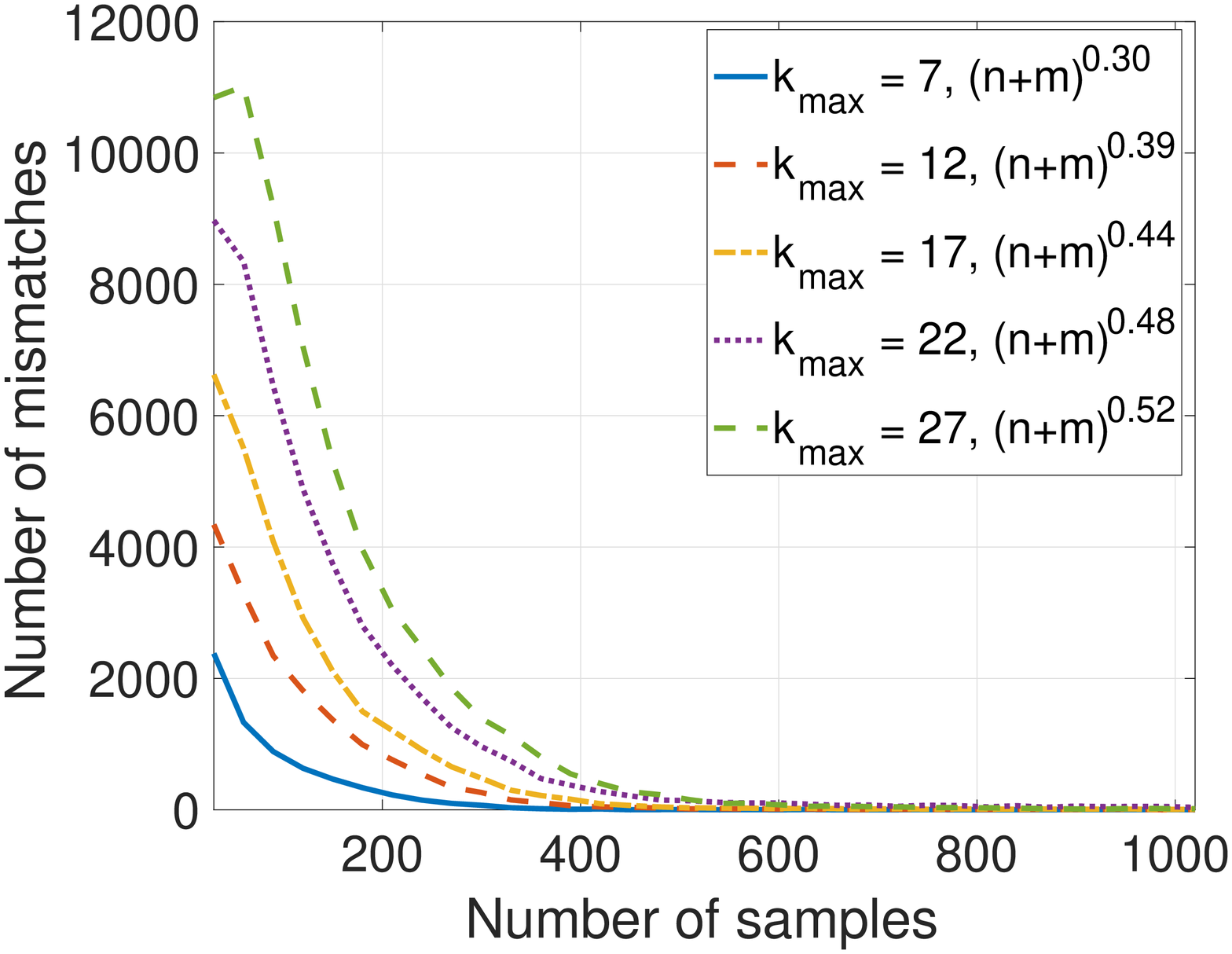}}
	\subfloat[Estimation error for different $k_{\max}$]{\label{fig_w_error}
		\includegraphics[width=.32\columnwidth]{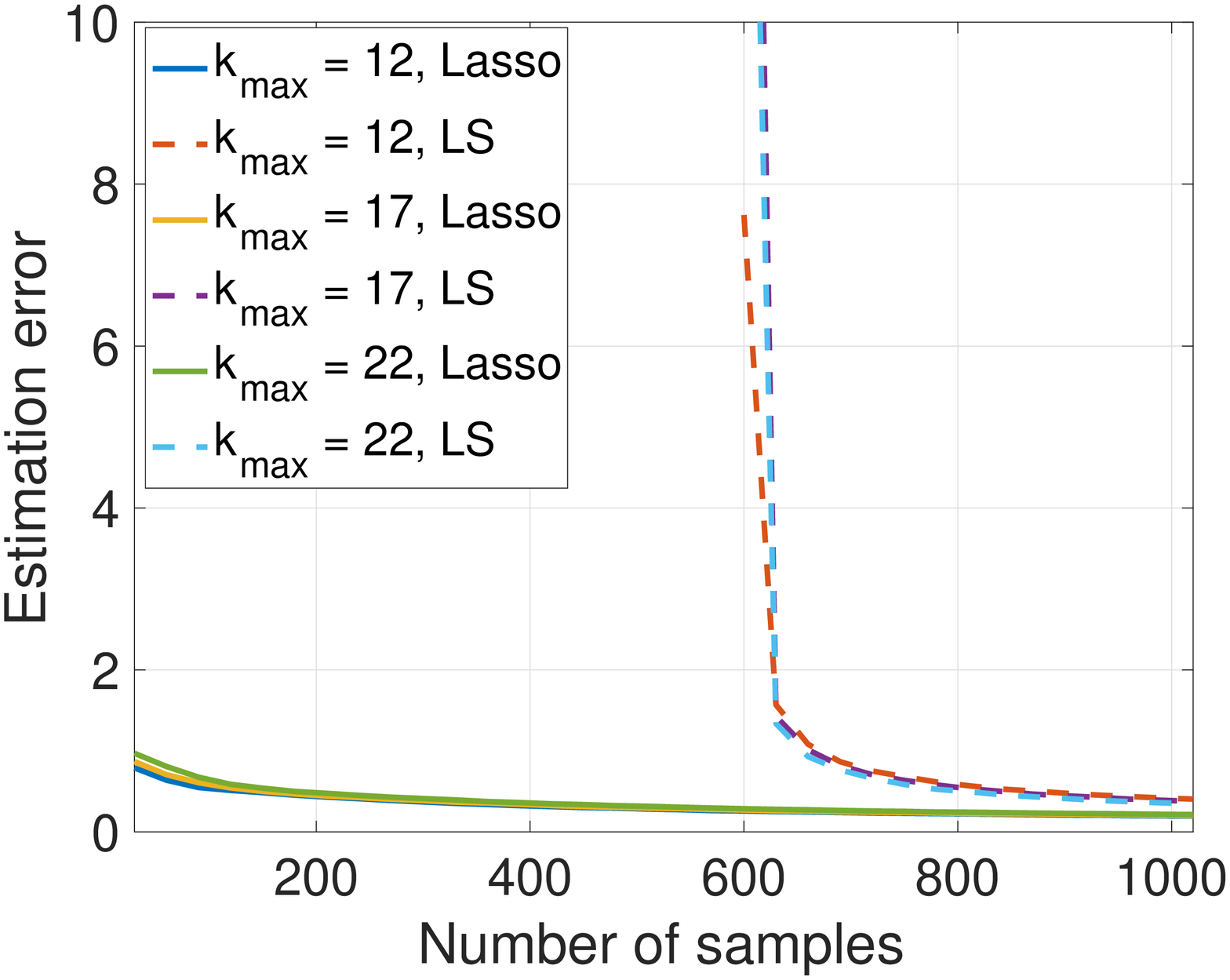}}
	\caption{ \footnotesize (a) The mismatch error with respect to the number of sample trajectories for different per-column number of nonzero elements in $\Theta^*$, (b) the normalized estimation error for Lasso and least-squares (abbreviated as LS) estimators with respect to the number of sample trajectories.}
\end{figure*}

\subsection{Case Study 2: Mass-Spring Systems}

\begin{figure}[t!]
	\centering
	\includegraphics[width=3in]{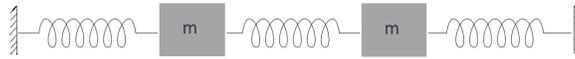}
	\caption[picture5]
	{ \footnotesize Mass-spring system with two masses}
	\label{fig_mass}
	\vspace{-1em}\end{figure}

In this case study, we conduct simulations on mass-spring networks with different sizes to elucidate the performance of the block-regularized estimator on physical systems. Consider $N$ identical masses connected via a path of springs. Figure~\ref{fig_mass} exemplifies this system for $N=2$. The state-space equation of this system in continuous domain can be written as
\begin{equation}
	\dot{x}_c(t) = A_cx(t)+ B_cu(t)
\end{equation}
where $A_c\in\mathbb{R}^{2N\times 2N}$, $B_c\in\mathbb{R}^{2N\times N}$, and $x_c(t)$ consists of two parts: the first $N$ elements correspond to the locations of the masses while the second $N$ elements capture their velocities.  With the unit masses and spring constants, we have
\begin{equation}
	A_c = \begin{bmatrix}
	0 & I\\
	S & 0
	\end{bmatrix},\quad B_c = \begin{bmatrix}
	0\\
	I
	\end{bmatrix}
\end{equation}
where $S\in\mathbb{R}^{n\times n}$ is a tridiagonal matrix whose diagonal elements are set to $-2$ and its first upper and lower diagonal elements are set to $1$~\cite{lin2013design}. This continuous system is discretized using the forward Euler method with sampling time of $0.2$ seconds. Similar to the previous case, $\Sigma_u$ and $\Sigma_w$ are set to $I$ and $0.5I$, respectively. Furthermore, $T$ is set to $3$ and $\lambda_d$ is chosen as~\eqref{lambda3}. The mutual incoherence property holds in all of these simulations. Notice that RST is equal to $3d/N$ in this case study, since $n = 2N$ and $m = N$. Figure~\ref{fig_mas_reld} depicts the required RST to achieve $\text{RME}\leq 0.1\%$ for different numbers of masses. If $N = 30$, the minimum number of sample trajectories to guarantee $\text{RME}\leq 0.1\%$ is $5$ times higher than the dimension of the system, whereas this number is dropped to $0.02$ for $N = 500$. Furthermore, Figure~\ref{fig_mas_error} shows the normalized estimation errors of the block-regularized and least-squares estimators for $N = 300$ and $T = 3$. Similar to the previous case study, the proposed estimator results in significantly smaller estimation errors in medium- and large-sampling regimes, while it remains as the only viable estimator when the number of available sample trajectories is smaller than $900$.

\begin{figure*}
	\centering
	\subfloat[Relative mismatch error]{\label{fig_mas_reld}
		\includegraphics[width=.32\columnwidth]{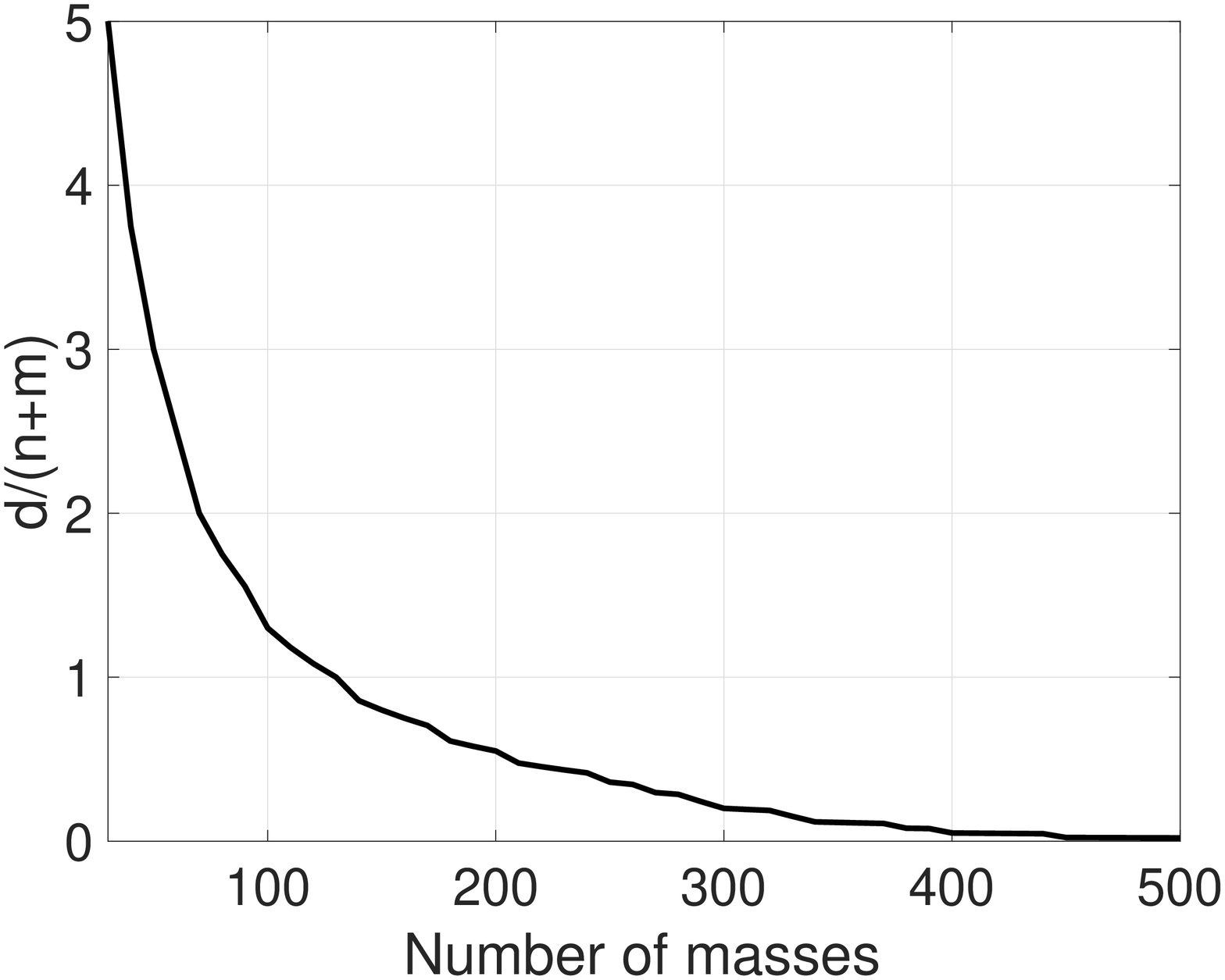}}
	\subfloat[Normalized estimation error]{\label{fig_mas_error}
		\includegraphics[width=.32\columnwidth]{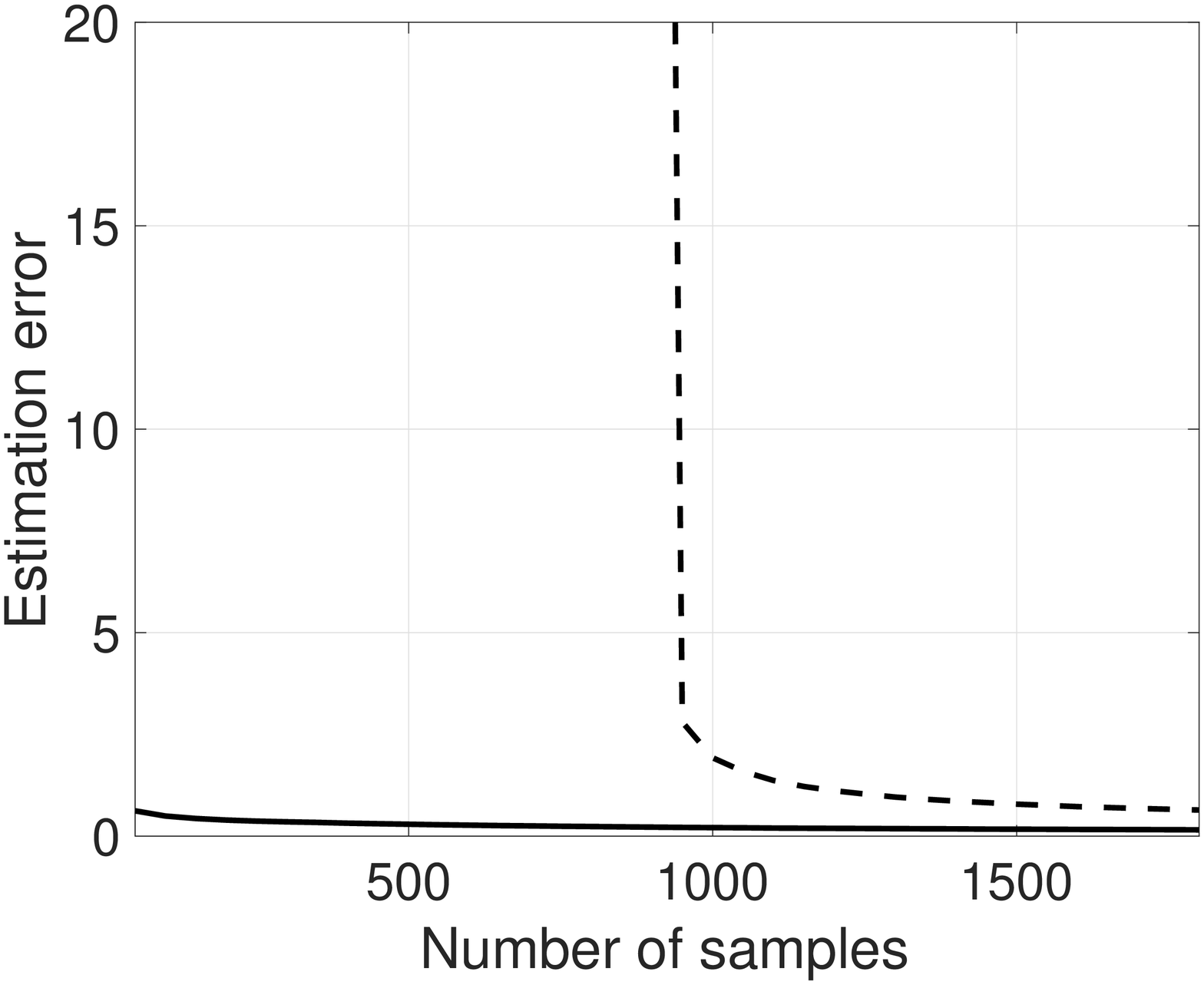}}
	\caption{ \footnotesize (a) The minimum RST to guarantee $\text{RME}\leq 0.1\%$ with respect to the number of masses, (b) the normalized estimation error for Lasso and least-squares estimators with respect to the number of sample trajectories.}
\end{figure*}

\subsection{Case Study 3: Switching Networks}
In this case study, we study a network of multi-agent systems that are interconnected through a switching information exchange topology. Recently, a special attention has been devoted to multi-agent systems with a time-varying network topology; in many communication networks, each sensor  has access only to the information of its neighbors. Therefore, when the location of these sensors changes over time, so does the topology of the interconnecting links~\cite{mesbahi2010graph}. The \textit{dwell time} is defined as the time interval in which the network topology is unchanged. The goal is to identify the structure of the network within the dwell time. The state-space equation of agent $i$ admits the following general form:
\begin{equation}
	\dot{x}_i(t) = \sum_{(i,j)\in\mathcal{N}_x(i)}A^{(i,j)}x_j(t) + \sum_{(i,j)\in\mathcal{N}_u(i)}B^{(i,j)}u_j(t)+w_i(t)
\end{equation}
where, as before, $A^{(i,j)}\in\mathbb{R}^{n_i\times n_i}$ and $B^{(i,j)}\in\mathbb{R}^{n_i\times m_i}$ are the $(i,j)^{\text{th}}$ blocks of $A$ and $B$. Furthermore, $\mathcal{N}_x(i)$ and $\mathcal{N}_u(i)$ are the sets of neighbors of agent $i$ whose respective state and input actions affect the state of agent $i$. 

We consider 200 agents connected through a randomly generated sparse network. In particular, we assume that each agent is connected to 5 other agents. If $j\in\mathcal{N}_x(i)$ or $j\in\mathcal{N}_u(i)$, then each element of $A^{(i,j)}$ or $B^{(i,j)}$ is randomly selected from $[-0.4\ -0.3]\cup [0.3\ 0.4]$. The behavior of the proposed block-regularized estimator will be examined for different dimensions of the agents. In particular, $(n_i, m_i)$ will be chosen from $\{(5,5), (8,8), (11,11)\}$. This entails that $D\in\{25, 64, 121\}$ and $(n,m) \in \{(1000,1000), (1600,1600), (2200,2200)\}$. Similar to the previous case study, $T$ is set to $3$ and the system is discretized using the forward Euler method with the sampling time of $0.2$ seconds. This implies that each sample trajectory is collected within $0.6$ seconds. The respective number of block mismatch and normalized estimation errors is depicted in Figures~\ref{fig_switch_mismatch} and~\ref{fig_switch_norm} with respect to the dwell time. It can be seen that as the dwell time becomes longer, a larger number of sample trajectories can be obtained, which in turn results in smaller block mismatch and normalized estimation errors. Furthermore, as the size of the blocks grows, an accurate identification of the system parameters requires more sample trajectories. In particular $195$, $278$, and $434$ sample trajectories are needed to achieve $\text{RME}\leq 0.1\%$ when $D$ is equal to $25$, $64$, and $121$, respectively. In contrast, the least-squares estimator requires at least $2000$, $3200$, and $4400$ sample trajectories to be uniquely defined. 

\begin{figure*}
	\centering
	\subfloat[Mismatch error]{\label{fig_switch_mismatch}
		\includegraphics[width=.32\columnwidth]{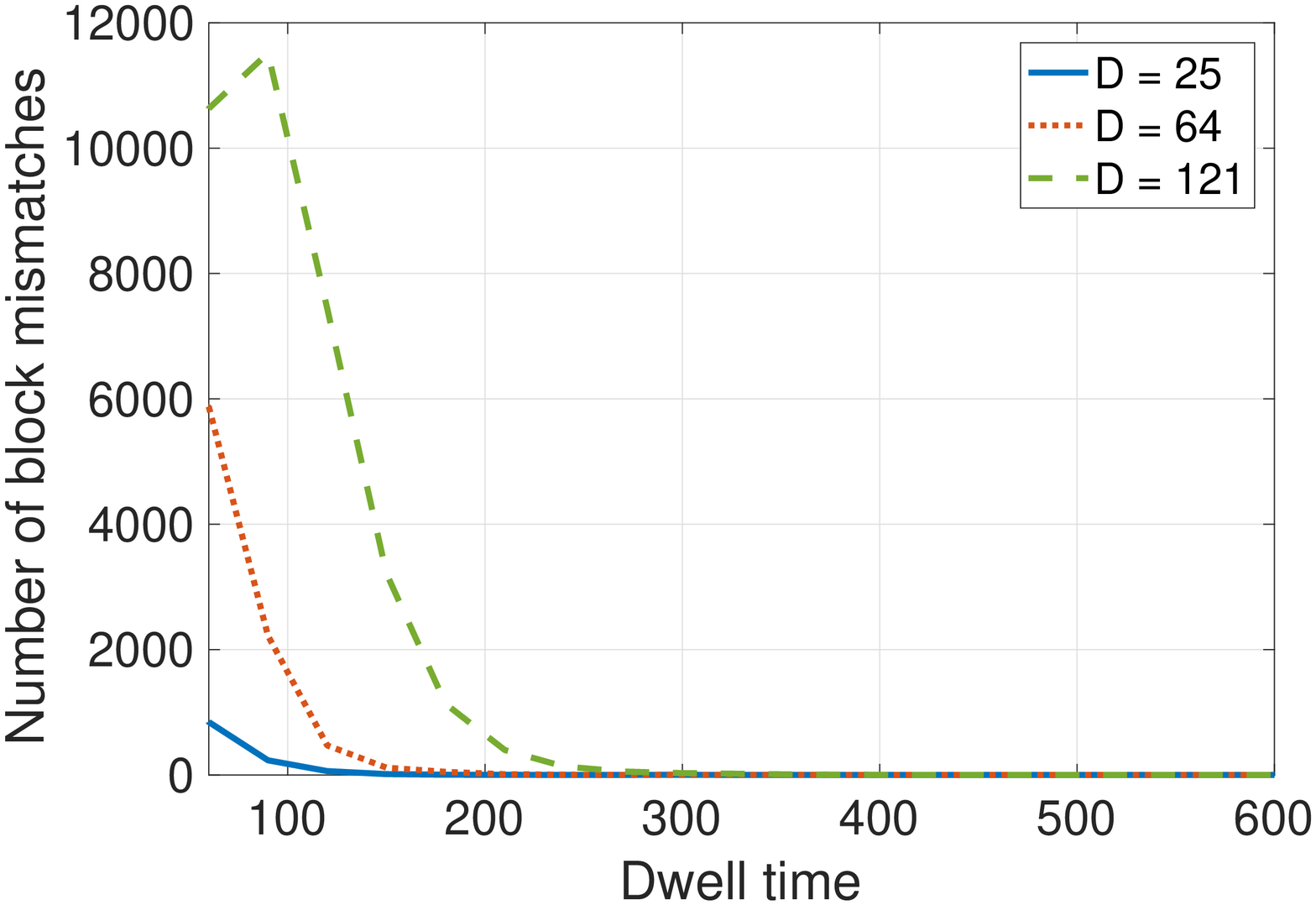}}
	\subfloat[Normalized estimation error]{\label{fig_switch_norm}
		\includegraphics[width=.32\columnwidth]{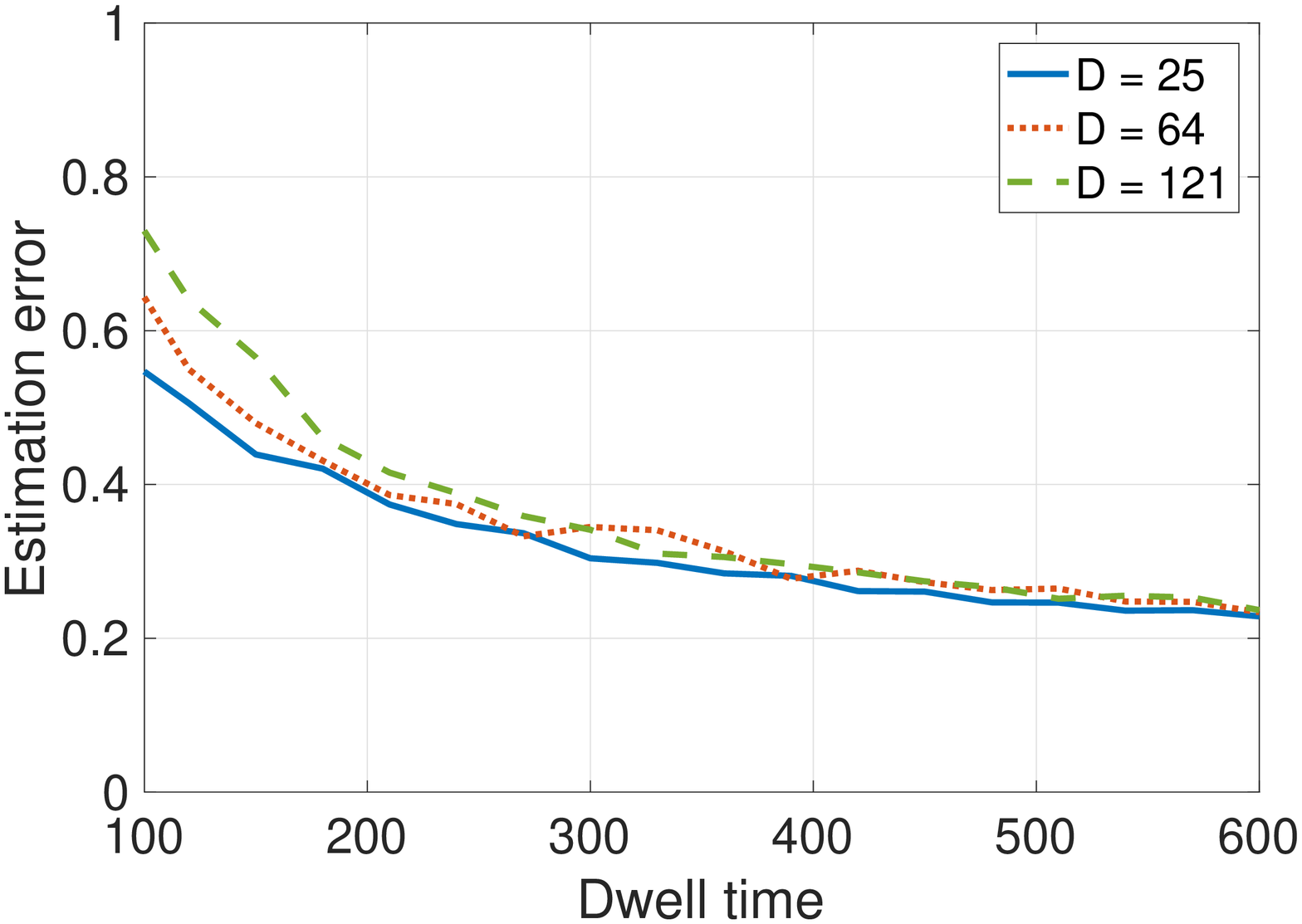}}
	\caption{ \footnotesize (a) The block mismatch error with respect to the dwell time for different block sizes in $\Theta^*$, (b) the normalized estimation error with respect to the dwell time for different block sizes in $\Theta^*$.}
\end{figure*}

\section{Conclusion}
We consider the problem of identifying the parameters of linear time-invariant (LTI) systems. In many real-world problems, the state-space equation describing the evolution of the system admits a block-sparse representation due to localized or internally limited interactions of its states and inputs. In this work, we leverage this property and introduce a block-regularized estimator to identify the sparse representation of the system. Using modern high-dimensional statistics, we derive sharp non-asymptotic bounds on the minimum number of input-state data samples to guarantee a small element-wise estimation error. In particular, we show that the number of available sample trajectories can be significantly smaller than the system dimension and yet, the proposed block-regularized estimator can correctly recover the block-sparsity of the state and input matrices and result in a small element-wise error. Through different case studies on synthetically generated systems, mass-spring networks, and multi-agent systems, we demonstrate substantial improvements in the accuracy of the proposed estimator, compared to its well-known least-squares counterpart.


\bibliographystyle{IEEEtran}
\bibliography{bib.bib}

\end{document}